\documentclass[a4paper,11pt]{article}

\usepackage{fullpage}
\usepackage{times}
\usepackage{soul}
\usepackage{url}
\usepackage[utf8]{inputenc}
\usepackage[small]{caption}
\usepackage{graphicx}
\usepackage[table]{xcolor}
\usepackage{amsmath}
\usepackage{booktabs}
\usepackage{algorithm}
\usepackage{dsfont}
\usepackage{enumerate}
\usepackage{tikz}
\usepackage{cleveref}
\usepackage{listings}

\usepackage{amsthm}
\usepackage{amssymb}
\usepackage{algorithmic}
\usepackage{bbm}

\newtheorem{theorem}{Theorem}[section]
\newtheorem{corollary}[theorem]{Corollary}
\newtheorem{proposition}[theorem]{Proposition}
\newtheorem{lemma}[theorem]{Lemma}

\theoremstyle{definition}
\newtheorem{definition}[theorem]{Definition}

\DeclareMathOperator*{\argmax}{arg\,max}
\DeclareMathOperator*{\argmin}{arg\,min}

\newcommand*\circled[1]{\tikz[baseline=(char.base)]{
            \node[shape=circle,draw,inner sep=2pt] (char) {#1};}}
            
\usepackage{natbib}

\usepackage{todonotes}

\newcommand{\gp}[1]

\allowdisplaybreaks

\usepackage{authblk}

\title{\bf Temporal Fair Division of Indivisible Items}

\author[1,2]{Edith Elkind}
\author[3]{Alexander Lam}
\author[4]{Mohamad Latifian}
\author[5]{Tzeh Yuan Neoh}
\author[1]{Nicholas Teh}

\affil[1]{University of Oxford, UK}
\affil[2]{Alan Turing Institute, UK}
\affil[3]{City University of Hong Kong, Hong Kong SAR}
\affil[4]{University of Edinburgh, UK}
\affil[5]{Agency for Science, Technology and Research, Singapore}

\date{\vspace{-10mm}}
\begin{document}
\maketitle

\begin{abstract}
    We study a fair division model where indivisible items arrive sequentially, and must be allocated immediately and irrevocably. 
    Previous work on online fair division has shown impossibility results in achieving approximate envy-freeness under these constraints. 
    In contrast, we consider an informed setting where the algorithm has complete knowledge of future items, and aim to ensure that the cumulative allocation at each round satisfies approximate envy-freeness--which we define as \emph{temporal envy-freeness up to one item} (TEF1).
    We focus on settings where items can be exclusively goods or exclusively chores.
    For goods, while TEF1 allocations may not always exist, we identify several special cases where they do---two agents, two item types, generalized binary valuations, unimodal preferences---and provide polynomial-time algorithms for these cases. We also prove that determining the existence of a TEF1 allocation is NP-hard. For chores, we establish analogous results for the special cases, but present a slightly weaker intractability result.
    We also establish the incompatibility between TEF1 and Pareto-optimality, with the implication that it is intractable to find a TEF1 allocation that maximizes any $p$-mean welfare, even for two agents.
\end{abstract}

\section{Introduction}
Fair division, a topic at the intersection of economics and computer science, has been extensively studied over the years, with applications spanning from divorce settlements and inheritance disputes to load balancing \citep{BramsTa96,Moulin03}. A typical fair division problem consists of a set of agents and a set of items, and the goal is to obtain a fair allocation of items to agents. In our work, we study the model where these items are \emph{indivisible}, so each must be wholly allocated to an agent. Moreover, the items can either provide positive utility (also called \emph{goods}) or negative utility (also called \emph{chores} or \emph{tasks}). When allocating indivisible items, a desirable and widely-studied fairness notion is \emph{envy-freeness up to one item (EF1)}, a natural relaxation of \emph{envy-freeness (EF)}. In an envy-free allocation, each agent values their own received bundle of items at least as much as every other agent's bundle, but this is not always possible for indivisible items.\footnote{Consider the example of two agents and a single item which must be allocated.} 
Accordingly, in an EF1 allocation, any envy that an agent has towards another agent can be eliminated by removing a single item from the latter or former agent's bundle, depending on whether the items in question are goods or chores, respectively.

Most prior research studies the problem in the \emph{offline} setting, assuming that all of the items are immediately available and ready to be allocated. However, there are various applications where the items arrive and need to be allocated on the spot in a sequential manner. For example, when the university administration places an order for lab equipment, or when a company orders new machines for its franchises, the items may not arrive at the same time, and instead arrive over time due to their availabilities and delivery logistics.  In the case of chores, one may consider the division of tasks over time in collaborative project management.
For a variety of reasons, these items may need to be immediately allocated; there may not be any storage space to keep any unallocated goods,
or the central decision maker may desire a non-wasteful allocation in the sense that items or tasks should not sit idle for periods of time.

These applications can be captured by an \emph{online fair division} model, in which items arrive over time and must be immediately and irrevocably allocated, though it is assumed that each item's valuation is not known until its arrival. Prior research has found that a complete EF1 allocation  of goods cannot be guaranteed under the online fair division model\footnote{In fact, the maximum pairwise envy is $\Omega(\sqrt{t})$ after $t$ rounds, worsening over time.} \citep{BKP+24}. However, this result relies on the assumption that the algorithm has no information about the future. In contrast, in our examples, the delivery services could provide the estimated delivery dates, and there may be a pre-planned timeline for the tasks. Motivated by these nuances, our work studies the \emph{informed online fair division setting}, assuming that the algorithm can access the items' valuations and arrival order upfront.

Note that the assumption of complete information about the future trivially leads to a complete EF1 allocation at the \emph{end} of the allocation period: simply treat the instance as an offline problem and apply any algorithm known to satisfy EF1 (e.g., \citep{LMMS04,AzizCaIg22chores,CaragiannisKuMo19}). However, this approach ignores the cumulative bundles of the items throughout the allocation period, and, consequently, agents may feel that their partial allocations are unfair for extended periods of time. Inspired by this issue, we propose \emph{temporal EF1 (TEF1)}, an extension of EF1 to the informed online fair division setting which requires that at each round, the cumulative allocation satisfies EF1.

The main focus of our work is on achieving TEF1, so for the informed online fair division of indivisible goods or chores, we aim to answer the following existence and computational questions:
\begin{quote}
    \itshape
    Which restricted settings guarantee the existence of a TEF1 allocation, and can we compute such an allocation in polynomial time in these settings? Is it computationally tractable to determine the existence\footnote{Prior work by \citet{he2019fairerfuturepast} has shown that a TEF1 allocation is not guaranteed to exist for goods in the general setting with three or more agents.} of a TEF1 allocation? In terms of existence and tractability, is TEF1 compatible with natural notions of efficiency?
\end{quote}

\subsection{Related Work} 
Our work is closely related to \emph{online fair division}, whereby items arrive over time and must be irrevocably allocated to agents.
The key difference is that in the standard online setting, the algorithm has completely no information on future items, whereas we assume complete future information.
Moreover, the goal in online fair division models is typically
to guarantee a fair allocation to agents at the end of the time horizon, rather than at every round.

As we focus on EF1, papers satisfying envy-based notions in online allocations are particularly relevant. 
\citet{AAGW15} consider envy-freeness from both ex-ante and ex-post standpoints, giving a best-of-both-worlds style result by designing an algorithm for goods which is envy-free in expectation and guarantees a bounded level of envy-freeness. Additionally, \citet{BKP+24} find that allocating goods uniformly at random leads to maximum pairwise envy which is sublinear in the number of rounds. For further reading, we refer the reader to the surveys by \citet{AlWa20} and \citet{AAB+23}. 

There has also been work on online fair division with \textit{partial information} on future items. \citet{BKP+24} study the extent to which approximations of envy-freeness and Pareto efficiency can be simultaneously satisfied under a spectrum of information settings, ranging from identical agents and i.i.d. valuations to zero future information.
An emerging line of work on \emph{learning-augmented} online algorithms has an alternate approach to partial future information: the algorithms are aided by (possibly inaccurate) predictions, typically from a machine-learning algorithm. The focus is to design algorithms which perform \emph{consistently} well with accurate predictions, and are \emph{robust} under inaccurate predictions. 
These predictions could be of each agent's total utility for the entire item set \citep{banerjee2022nash_predictions,BGH+23}, or for a random subset of $k$ incoming items \citep{BePe24}.

Unlike the aforementioned papers, our work considers a completely \emph{informed} variant of the online fair division setting, which has been studied by \citet{he2019fairerfuturepast} for the allocation of goods. Similar to our paper, their objective is to ensure that EF1 is satisfied at each round, but they allow agents to swap their bundles. When multiple goods may arrive at each round, our setting also generalizes the \emph{repeated} fair division setting, in which the same set of goods arrives at each round. For this model, \citet{igarashi2023repeatedfairallocation} give results on the existence of allocations which are envy-free and Pareto optimal in the end, with the items in each round being allocated in an EF1 manner. However, they do not analyse whether the \emph{cumulative} allocation at each round can satisfy some fairness constraint, which is the focus of our paper.
\citet{caragiannis2024repeatedmatching} also consider a model where the same set of items appear at each round, but each agent gets exactly one item per round.

When the valuations are known upfront for the allocation of chores, the model is similar to the field of work on job scheduling. There have been numerous papers studying fair scheduling, but the fairness is typically represented by an objective function which the algorithm aims to minimize or approximate \citep{ScYa00,ImMo20,Baru95}. On the other hand, there is little work on satisfying envy-based notions in scheduling problems, but \citet{LLZ21} study the compatibility of EF1 and Pareto optimality in various settings.
While we consider separately the cases of goods allocation and chores allocation, to the best of our knowledge, there is no prior work which studies an online fair division model with both goods and chores in the same instance under any information assumption.

Temporal models that study concepts of achieving fairness over time have also been recently considered in the social choice literature \citep{chandak2024proportional,elkind2022temporalslot,elkind2024temporalsurvey,elkind2024temporalelections,mackenzie2023takingturns,lackner2020perpetual,zech2024multiwinnerchange}.

\subsection{Our Contributions}
We outline our paper's answers to these key questions as follows.

In Section \ref{sec:existence}, we show the existence of TEF1 allocations (for goods or chores) in restricted settings, such as the case of two agents, when there are two types of items, when agents have generalized binary valuations, or when they have unimodal preferences.
For each of these cases, we provide an accompanying polynomial-time algorithm.
For the allocation of goods, we show that determining whether there exists a TEF1 allocation is NP-hard; whereas for chores, we show that given a partial TEF1 allocation, it is NP-hard to determine if there exists a TEF1 allocation that allocates all remaining chores.

In Section \ref{sec:efficiency}, we investigate the compatibility of TEF1 and Pareto-optimality (PO).
We show that even in the case of two agents, while a TEF1 allocation is known to exist and can be computed in polynomial time (for both goods and chores), existence is no longer guaranteed if we mandate PO as well.
Moreover, we show that in this same setting, determining the existence of TEF1 and PO allocations is NP-hard.
Our result also directly implies the computational intractability of determining whether there exists a TEF1 allocation that maximizes any $p$-mean welfare objective (which subsumes most popular social welfare objectives).

Finally, in Section \ref{sec:multigoods}, we consider the special case where the same set of items arrive at each round, and show that even determining whether repeating a particular allocation in two consecutive rounds can result in a TEF1 allocation is NP-hard.
We complement this with a polynomial-time algorithm for computing a TEF1 allocation in this case when there are just two rounds.

\section{Preliminaries}
For each positive integer $k$, let $[k] := \{1,\dots,k\}$. 
We consider the problem of fairly allocating indivisible items to agents over multiple rounds. 
Let an instance of the informed online fair division problem be denoted by $\mathcal{I} = \langle N,T,\{O_t\}_{t\in [T]}, \mathbf{v}=(v_1,\dots,v_n)\rangle$, in which we have a set of agents $N = [n]$ and a set $O$ of $m$ items which arrive over $T$ rounds, and are to be allocated to the agents. For each $t\in [T]$, we denote the set of items that arrive at round $t$ by $O_t$, and define the cumulative set of items that arrived in rounds $1,\dots,t$ by $O^{t} := \bigcup_{\ell\in [t]} O_\ell$. Note that $O = O^{T}$. 

We assume that each agent $i \in N$ has an additive \emph{valuation function} $v_i: 2^O \to \mathbb{R}$ over the items, i.e., for $S \subseteq O$, $v_i(S) = \sum_{o \in S} v_i(\{o\})$. For notational convenience, we write $v_i(o)$ instead of $v_i(\{o\})$ for a single item $o \in O$, and $v$ instead of $v_i$ when valuation functions are identical. Denote $\mathbf{v} = (v_1,\ldots,v_n)$ as the \emph{valuation profile}. In this work, we consider two cases: \emph{goods allocation} where for each $i\in N$ and $o\in O$, $v_i(o)\geq 0$, and \emph{chores allocation} where for each $i\in N$ and $o\in O$, $v_i(o)\leq 0$. For clarity, we use $g$ instead of $o$ and refer to items as goods when explicitly referring to the goods setting, and $c$ instead of $o$ and refer to the items as chores when considering the chores setting.

An allocation $\mathcal{A} = (A_1, \ldots, A_n)$ of items in $O$ to the agents is an ordered partition of $O$, i.e. for $i, j \in N$, $A_i \cap A_j = \varnothing$ and $\bigcup_{i \in N} A_i = O$. In addition, for $t \in [T]$ we denote the allocation after round $t$ by $\mathcal{A}^{t} = (A^{t}_1, \ldots, A^{t}_n)$ where $A^{t}_i = A_i \cap O^{t}$.
For $t < T$, we sometimes refer to $\mathcal{A}^{t}$ as a \emph{partial allocation}. Note that $\mathcal{A} = \mathcal{A}^{T}$.

Our goal is to find an allocation that is fair after each round. The main fairness notion that we consider is \emph{envy-freeness up to one item (EF1)}, which is a well-studied notion in the problem of allocating indivisible items.

\begin{definition}
In a goods (resp. chores) allocation instance, an allocation $\mathcal{A} = (A_1, \ldots, A_n)$ is said to be EF1 if for all pairs of agents $i,j \in N$, 
there exists a good $g \in A_j$ (resp. chore $c \in A_i$) such that $v_i(A_i) \geq v_i(A_j \setminus \{g\})$ (resp. $v_i(A_i\setminus \{c\}) \geq v_i(A_j )$). 
\end{definition}
In this work we target fairness in a \emph{cumulative} sense, introducing the notion of \emph{temporal envy-freeness up to one item (TEF1)} which requires that at every round prefix, the cumulative allocation of items that have arrived so far satisfies EF1.
\begin{definition}[Temporal EF1]
    For any $t \in [T]$, an allocation $\mathcal{A}^t = (A^t_1, \dots, A^t_n)$ is said to be \emph{temporal envy-free up to one item (TEF1)} if for all $t' \leq t$, the allocation $\mathcal{A}^{t'}$ is EF1. 
\end{definition}
It is important to note the distinction between TEF1 and EF1: while the EF1 property applies only to a single allocation, TEF1 extends this requirement by ensuring that the cumulative allocation at every prior round also satisfies EF1.

However, the possible non-existence of TEF1 allocations has been shown in the general goods allocation setting, as
\citet[Thm. 4.2]{he2019fairerfuturepast} illustrated using a counterexample with $3$ agents and $23$ items, which can be generalized to $n>3$ agents. For completeness, we include this counterexample along with an intuitive explanation in Appendix~\ref{app:counterexample}.
We remark that this counterexample cannot be modified to show a similar non-existence result in the case of chores.
As such, while we conjecture that this non-existence result should also hold for chores, this remains an open question.

We assume that the reader is familiar with basic notions of classic complexity theory~\citep{papadimitriou_computational_2007}.
All omitted proofs can be found in Appendices \ref{app:omitted_existence} and \ref{app:omitted_efficiency}.

\section{On the Existence of TEF1 Allocations} \label{sec:existence}
We know that TEF1 does not always exist in general instances. 
However, there are some restricted classes of instances under which a TEF1 allocation is guaranteed to exist. 
In this section we identify several such settings, but before doing so, we first demonstrate that it typically suffices to prove results for instances where only one item appears at each round (i.e., $T=m$ and $|O_i| = 1$). 
Since having a single item per round is a special case of the setting with multiple items per round, any impossibility result for the former directly extends to the latter. 
Conversely, as we will show in the following lemma, any positive existence result for the single-item setting also implies a positive result for the general case with multiple items per round. 

\begin{lemma}\label{lem:transform}
    Let $\mathcal{I}$ be an instance with $T$ rounds and a total of $m$ items, where multiple items can appear in a single round. We can transform $\mathcal{I}$ into an equivalent instance $\mathcal{I}^{=1}$ with exactly $m$ rounds, such that each round has exactly one item.
    Then, if a \emph{TEF1} allocation exists for $\mathcal{I}^{=1}$, a \emph{TEF1} allocation also exists for $\mathcal{I}$. 
\end{lemma}
\begin{proof}
    Consider an arbitrary instance $\mathcal{I} = \langle N, T, \{O_t\}_{t\in [T]}, \mathbf{v} = (v_1,\dots, \allowbreak v_n)\rangle$.
    Without loss of generality, we assume that items are labelled in a non-decreasing fashion (with respect to the rounds), i.e., for any two rounds $r,r' \in [T]$ where $ r< r'$, if $o_j \in O_r$ and $o_{j'} \in O_{r'}$, then $j < j'$.  
    We first transform $\mathcal{I}$ into another instance $\mathcal{I}^{=1} = \langle N,m, \{\widetilde{O}_t\}_{t\in [m]}, \mathbf{v}\rangle$ where $\widetilde{O}_t = \{o_t\}$ for each $t \in [m]$.
    Suppose there exists a TEF1 allocation $\mathcal{A}$ for $\mathcal{I}^{=1}$, i.e., for each $t \in [m]$, the partial allocation $\mathcal{A}^t = (A^t_1,\dots,A^t_n)$ (where $A^t_i = A_i \cap \widetilde{O}^t$ for each $i \in N$) is EF1.
    We construct an allocation $\mathcal{B}$ for instance $\mathcal{I}$.
    For each $t \in [T]$, let $t' \in [m]$ be the round such that $O^t = \widetilde{O}^{t'}$, and set $\mathcal{B}^t = \mathcal{A}^{t'}$. 
    Note that $\{O^t\}_{t \in [T]} \subseteq \{\widetilde{O}^t\}_{t\in [m]}$, so since $\mathcal{A}$ is a TEF1 allocation for $\mathcal{I}^{=1}$, $\mathcal{B}$ is a TEF1 allocation for $\mathcal{I}$. 
\end{proof}
As we will show later, the above result cannot be immediately applied to settings where preferences evolve over time (e.g., for unimodal preferences), but this is due to the inconsistency of these settings between the single and multiple items cases. For the remainder of the paper, we simplify the notation based on the transformation in the proof of Lemma~\ref{lem:transform}: when a single item arrives at each round, the item subscript indicates the round in which it arrives (i.e. the single item in $O_i$ is denoted by $o_i$ or $g_i$ or $c_i$).

\subsection{Two Agents} \label{sec:twoagents}

\citet[Thm. 3.4]{he2019fairerfuturepast} showed the existence of a TEF1 allocation for goods when $n=2$, by presenting a polynomial-time algorithm that returns such an allocation.
We begin by extending this result to the case of chores.

Intuitively, in each round the algorithm greedily allocates the single chore that arrives in that round to an agent that does not envy the other agent in the current (partial) allocation.
A counter $s$ is used to denote the last round after which $\mathcal{A}^s$ was envy-free, and if the allocation of a chore $c_t$ (for some round $t \in [m]$) results in both agents envying each other in $\mathcal{A}^t \setminus \mathcal{A}^s$, then the agents' bundles in $\mathcal{A}^t \setminus \mathcal{A}^s$ are swapped.

\begin{algorithm}[h!]
   \caption{Returns a TEF1 allocation for chores when $n=2$}
   \label{alg:twoagents_chores}
   \begin{flushleft} \textbf{Input}: Set of agents $N=\{1,\dots,n\}$, set of chores $O=\{c_1,\dots,c_m\}$, and valuation profile $\mathbf{v} = (v_1, v_2)$ \\
   \textbf{Output}: TEF1 allocation $\mathcal{A}$ of chores in $O$ to agents in $N$ 
   \end{flushleft}
   \begin{algorithmic}[1]
       \STATE Initialize $s\leftarrow 0$ and  $\mathcal{A}^0 \leftarrow (\varnothing,\varnothing)$
       \FOR{$t = 1,2,\dots,m$}
       \IF{$v_1(A^{t-1}_1 \setminus A^s_1)\geq v_1(A^{t-1}_2 \setminus A^s_2)$}
           \STATE $\mathcal{A}^t\leftarrow (A_1^{t-1}\cup \{c_t\}, A_2^{t-1})$
       \ELSE
           \STATE $\mathcal{A}^t\leftarrow (A_1^{t-1}, A_2^{t-1}\cup \{c_t\})$
       \ENDIF
       \IF{$v_1(A^t_1 \setminus A^s_1) < v_1(A^t_2 \setminus A^s_2)$ and $v_2(A^t_2 \setminus A^s_2) < v_2(A^t_1 \setminus A^t_1)$}
           \STATE $\mathcal{A}^t\leftarrow (A_1^s\cup A_2^t \setminus A^s_2, A_2^s\cup A_1^t \setminus A_1^s)$
       \ENDIF
       \IF{$v_1(A^t_1 \setminus A^s_1)\geq v_1(A^t_2 \setminus A^s_2)$ and $v_2(A^t_2 \setminus A^s_2)\geq v_2(A^t_1 \setminus A^s_1)$}
           \STATE $s\leftarrow t$
       \ENDIF
       \ENDFOR
       \STATE \textbf{return} $\mathcal{A} = (A_1^m, A_2^m)$
   \end{algorithmic}
\end{algorithm}

\begin{theorem}\label{thm:2agents}
   For $n=2$, \Cref{alg:twoagents_chores} returns a \emph{TEF1} allocation for chores, and it runs in polynomial time.
\end{theorem}
\begin{proof}
    The polynomial runtime of \Cref{alg:twoagents_chores} is easy to verify, given there is only one \textbf{for} loop which runs in $\mathcal{O}(m)$, and the other operations within run $\mathcal{O}(m)$ time.
    Thus, we focus on proving correctness.
    
    For any $t \in [m]$, let $r_t < t$ be defined as the latest round before $t$ such that $\mathcal{A}^{r_t}$ is EF.
    This implies that if $\mathcal{A}^{t'} \setminus \mathcal{A}^{r_t}$ is EF1 for all $t'\in \{r_t,r_t+1,\dots,t\}$, then $\mathcal{A}^t$ is also TEF1. Therefore, it suffices to show that for every $t\in [m]$, $\mathcal{A}^t \setminus \mathcal{A}^{r_t}$ is EF1, which we will prove by induction.
    For the base case where $t=r_t+1$, $\mathcal{A}^t \setminus \mathcal{A}^{r_t}$ only consists of a single allocated chore $c_{r_t+1}$, and is trivially EF1.

    Next, we prove the inductive step. 
    Assume that for some $t > r_t+1$, $\mathcal{A}^{t-1} \setminus \mathcal{A}^{r_t}$ is EF1.
    We will show that $\mathcal{A}^{t} \setminus \mathcal{A}^{r_t}$ is also EF1. Let $r_t'$ be the earliest round ahead of $r_t$ such that $\mathcal{A}^{r_t'} \setminus \mathcal{A}^{r_t}$ is EF (if such a round exists). We divide the remainder of the proof into two cases depending on whether a partial bundle swap (as in line 9 of the algorithm) occurs at round $r_t'$.
    \begin{description}
        \item[Case 1: Round $r_t'$ does not exist or no swap at round $r'_t$.] 
        Suppose without loss of generality that $v_1(A^{t-1}_1 \setminus A^{r_t}_1)< v_1(A^{t-1}_2 \setminus A^{r_t}_2)$, i.e., agent $1$ envies agent $2$ in $\mathcal{A}^{t-1} \setminus \mathcal{A}^{r_t}$. Then agent $2$ does not envy agent $1$ (otherwise we would swap the bundles, contradicting the definition of $r_t$), and consequently receives $c_t$. If agent $2$ envies agent $1$ after receiving $c_t$, this envy can be removed by removing $c_t$. We also know that $\mathcal{A}^{t}\setminus \mathcal{A}^{r_t}$ is EF1 w.r.t. agent $1$ (who did not receive a chore in round $t$) because by our inductive assumption, $\mathcal{A}^{t-1} \setminus \mathcal{A}^{r_t}$ is EF1, concluding the proof of this case.
        \item[Case 2: Swap occurs at round $r_t'$.]
        We assume that $r_t'>t$, because if $r_t'=t$, then $\mathcal{A}^t\setminus \mathcal{A}^{r_t}$ is EF and therefore EF1. For each $i\in \{t-1,t\}$, let $\mathcal{B}^i\setminus \mathcal{A}^s$ refer to the algorithm's allocation of the chores $O^i\setminus O^s$ \emph{before} the bundle swap, and suppose without loss of generality that $v_1(B^{t-1}_1\setminus B^{r_t}_1)< v_1(B^{t-1}_2\setminus B^{r_t}_2)$. We therefore must have $v_2(B^{t-1}_2\setminus B^{r_t}_2)\geq v_2(B^{t-1}_1\setminus B^{r_t}_1)$ to avoid contradicting the definition of $r_t$. Since agent $2$ is not envied by agent $1$ in round $t-1$, it receives chore $c_t$, so we have $\mathcal{B}^t\setminus \mathcal{B}^{r_t}=(B^{t-1}_1\setminus B^{r_t}_1, (B^{t-1}_2\setminus B^{r_t}_2) \cup \{c_t\})$. This means that after the bundle swap is executed, we have $\mathcal{A}^t\setminus \mathcal{A}^{r_t}=((B^{t-1}_2\setminus B^{r_t}_2) \cup \{c_t\}, B^{t-1}_1\setminus B^{r_t}_1)$. Recall that $v_1(B^{t-1}_2\setminus B^{r_t}_2)>v_1(B^{t-1}_1\setminus B^{r_t}_1)$, so $c_t$ can be removed from agent $1$'s bundle to eliminate their envy towards agent $2$. Also, by the inductive assumption, there exists a chore $c\in A^{t-1}_2\setminus A^{r_t}_2$ such that $v_2((A^{t-1}_2\setminus A^{r_t}_2)\setminus \{c\})\geq v_2(A^{t-1}_1\setminus A^{r_t}_1)$. Observe that $A^{t}_2\setminus A^{r_t}_2=A^{t-1}_2\setminus A^{r_t}_2$ and $A^{t}_1\setminus A^{r_t}_1=(A^{t-1}_1\setminus A^{r_t}_1)\cup \{c_t\}$. Combining this with the inductive assumption, we have that there exists a chore $c\in A_2^t\setminus A_2^{r_t}$ such that
    \begin{align*}
    v_2((A^{t}_2\setminus A^{r_t}_2)\setminus \{c\})&=v_2((A^{t-1}_2\setminus A^{r_t}_2)\setminus \{c\})\\
    &\geq v_2(A^{t-1}_1\setminus A^{r_t}_1)\\
    &\geq v_2((A^{t-1}_1\setminus A^{r_t}_1)\cup \{c_t\})\\
    &=v_2(A^{t}_1\setminus A^{r_t}_1).
    \end{align*}
    Therefore, $\mathcal{A}^{t} \setminus \mathcal{A}^{r_t}$ is EF1 in this case.
    \end{description}
    We have shown that $\mathcal{A}^{t} \setminus \mathcal{A}^{r_t}$ is EF1 regardless of whether the allocation has undergone a bundle swap, so by induction, \Cref{alg:twoagents_chores} returns a TEF1 allocation for chores.
\end{proof}
Next, we consider \emph{temporal envy-freeness up to any good (TEFX)}, the temporal variant of the stronger \emph{envy-freeness up to any good (EFX)} notion (which is the same as EF1, except that we drop \emph{any} good instead of \emph{some} good).

\begin{definition}
    In a goods (resp. chores) allocation instance, an allocation $\mathcal{A} = (A_1, \ldots, A_n)$ is said to be EFX if for all pairs of agents $i,j \in N$, 
and all goods $g \in A_j$ (resp. chores $c \in A_i$) we have $v_i(A_i) \geq v_i(A_j \setminus \{g\})$ (resp. $v_i(A_i\setminus \{c\}) \geq v_i(A_j )$). 
\end{definition}

\begin{definition}[Temporal EFX]
    For any $t \in [T]$, an allocation $\mathcal{A}^t = (A^t_1, \dots, A^t_n)$ is said to be \emph{temporal envy-free up to any item (TEFX)} if for all $t' \leq t$, the allocation $\mathcal{A}^{t'}$ is EFX. 
\end{definition}

We then show that TEFX allocations (for goods or chores) may not exist, even for two agents with identical valuations, and when there are only two types of items.

\begin{proposition} \label{prop:tefx_n=2_notexist}
    A \emph{TEFX} allocation for goods or chores may not exist, even for $n=2$ with identical valuations and two types of items.
\end{proposition}
\begin{proof}
    We first prove the result for the case of goods.
    Consider the instance with two agents $N =  \{1,2\}$ and three goods $O = \{g_1,g_2,g_3\}$, where agents have identical valuations: $v(g_1) = v(g_2) = 1$ and $v(g_3) = 2$.
    In order for the partial allocation at the end of the second round to be TEFX, each agent must be allocated exactly one of $\{g_1,g_2\}$---suppose that agent~$1$ is allocated $g_1$ and agent~$2$ is allocated $g_2$.
    In the third round,
    without loss of generality, suppose that $g_3$ is allocated to agent~$1$.
    Then, agent~$2$ will still envy agent~$1$ even after dropping $g_1$ from agent~$1$'s bundle, as $v(A_2) = v(g_2) = 1 < v(A_1 \setminus \{g_1\}) = v(g_3) = 2$.

    Next, we prove the result for chores.
    Consider the instance with two agents $N =  \{1,2\}$ and three chores $O = \{c_1,c_2,c_3\}$, where agents have identical valuations: $v(c_1) = v(c_2) = -1$ and $v(c_3) = -2$.
    In order for the partial allocation at the end of the secound round to be TEFX, each agent must be allocated exactly one of $\{c_1,c_2\}$---suppose that agent~$1$ is allocated $c_1$ and agent~$2$ is allocated $c_2$.
    In the third round, without loss of generality, suppose that $c_3$ is allocated to agent~$1$.
    Then, agent~$1$ will still envy agent~$2$ even after dropping $c_1$ from her own bundle, as $v(A_1 \setminus \{c_1\}) = v(c_3) = -2 < v(A_2) = v(c_2) = -1$.
\end{proof}

\subsection{Other Restricted Settings}

The next natural question we ask is whether there are other special cases where EF1 allocation is guaranteed to exist.
We answer this question affirmatively by demonstrating the existence of EF1 allocations in three special cases, each supported by a polynomial-time algorithm that returns such an allocation.

\subsubsection{\textbf{Two Types of Items}}
The first setting we consider is one where items can be divided into two \emph{types}, and each agent values all items of a particular type equally.
Formally, let $S_1, S_2 \subseteq O$ be a partition of the set of items, so that $S_1 \cap S_2 = \varnothing$, and $S_1 \cup S_2 = O$. 
Then, for any $r \in \{1,2\}$, two items $o,o' \in S_r$, and agent $i \in N$, we have that $v_i(o) = v_i(o')$.

Settings with only two types of items/tasks arise naturally in various applications, such as distributing food and clothing donations from a charity, or allocating cleaning and cooking chores in a household.
This preference restriction has been studied for chores in offline settings \citep{ALRS23,GMQ24}, and we remark that agents may have distinct valuations for up to $2n$ different items, unlike the extensively studied \emph{bi-valued} preferences \citep{EPS22,GMQ22} which involve only two distinct item values.

We show that for this setting, a TEF1 allocation for goods or chores always exists and can be computed in polynomial time. 
Intuitively, the algorithm treats the two item types independently: items of the first type are allocated in a round-robin manner from agent $1$ to $n$, while items of the second type are allocated in reverse round-robin order from agent $n$ to $1$. 
Then, our result is as follows.
\begin{theorem} \label{thm:twotypes}
    When there are two types of items, 
    a \emph{TEF1} allocation for goods or chores exists and can be computed in polynomial time.
\end{theorem}

\subsubsection{\textbf{Generalized Binary Valuations}} 
The next setting we consider is one where agents have \emph{generalized binary valuations} (also known as \emph{restricted additive valuations} \citep{AkramiReSe22,CamachoFePe23}).
This class of valuation functions generalizes both identical and binary valuations, which are both widely studied in fair division 
\citep{halpern2020binary,plaut2020almost,Suksompong2022}.
Formally, we say that agents have \emph{generalized binary valuations} if for every agent $i\in N$ and item $o_j\in O$, $v_i(o_j)\in \{0,p_j\}$, where $p_j\in \mathbb{R}\setminus \{0\}$.

We show that for this setting, a TEF1 allocation can be computed efficiently, with the following result. We remark that the resulting allocation also satisfies \emph{Pareto-optimality} (\Cref{def:po}).
\begin{theorem} \label{thm:generalizedbinary}
    When agents have generalized binary valuations, a \emph{TEF1} allocation for goods or chores exists and can be computed in polynomial time.
\end{theorem}

\subsubsection{\textbf{Unimodal Preferences}}
The last setting that we consider is the class of \emph{unimodal preferences}, which consists of the widely studied \emph{single-peaked} and \emph{single-dipped} preference structures in social choice \citep{black1948singlepeaked,arrow2012socialchoice} and cake cutting \citep{Thom94,BKO20}. 
We adapt these concepts for the online fair division setting with a single item at each timestep. 

\begin{definition}
A valuation profile $\mathbf{v}$ is \emph{single-peaked} if for each agent $i \in N$, there is an item $o_{i^*}$ where for each $j,k \in [m]$ such that $j<k<i^*$, $v_i(o_j)\leq v_i(o_k)\leq v_i(o_{i^*})$, and for each $j,k \in [m]$ such that $i^*<j<k$, $v_i(o_{i^*})\geq v_i(o_j)\geq v_i(o_k)$.
\end{definition}
\begin{definition}
A valuation profile $\mathbf{v}$ is  \emph{single-dipped} if for each agent $i \in N$, there is an item $o_{i^*}$ where for each $j,k \in [m]$ such that $j<k<i^*$, $v_i(o_j)\geq v_i(o_k)\geq v_i(o_{i^*})$, and for each $j,k \in [m]$ such that $i^*<j<k$, $v_i(o_{i^*})\leq v_i(o_j)\leq v_i(o_k)$.
\end{definition}

In other words, under single-peaked (resp. single-dipped) valuations, agents have a specific item $o_{i^*}$ that they prefer (resp. dislike) the most, and prefer (resp. dislike) items less as they arrive further away in time from $o_{i^*}$.

Note that this restricted preference structure is well-defined for the setting of a single item arriving per round, but may not be compatible with a generalization to multiple items per round as described in Lemma~\ref{lem:transform} (unless the items in each round are identically-valued by agents).\footnote{Specifically, in the multiple items per round case, if the bundles of items at each timestep are unimodally valued, the single-item per round transformation of the instance may not necessarily be unimodal.}

Unimodal preferences may arise in settings where agents place higher value on resources at the time surrounding specific events. For example, in disaster relief, the demand for food and essential supplies peaks as a natural disaster approaches, then declines once the immediate crisis passes. 
Similarly, in project management, the workload for team members intensifies (in terms of required time and effort) as the project nears its deadline, but significantly decreases during the final stages, such as editing and proofreading.

Unimodal preferences also generalizes other standard preference restrictions studied in fair division and voting models, such as settings where agents have \emph{monotonic valuations} \citep{elkind2024verifying} or \emph{identical rankings} \citep{plaut2020almost}.

We propose efficient algorithms for computing a TEF1 allocation for goods when agents have single-peaked valuations, and for chores when agents have single-dipped valuations.

\begin{theorem} \label{thm:singlepeaked_goods}
    When agents have single-peaked valuations, a \emph{TEF1} allocation for goods exists and can be computed in polynomial time. When agents have single-dipped valuations, a \emph{TEF1} allocation for chores exists and can be computed in polynomial time.
\end{theorem}

We note that while a simple greedy algorithm performs well in the case of single-peaked valuations for goods and single-dipped valuations for chores, it fails in the reverse scenario—single-dipped valuations for goods and single-peaked valuations for chores. This is due to the fact that, in the latter case, the position of the dip or peak becomes critical and significantly complicates the way we allocate the item. We leave the existence of polynomial-time algorithm(s) for the reverse scenario as an open question.

\subsection{Hardness Results for TEF1 Allocations}
The non-existence of TEF1 goods allocations for $n \geq 3$ prompts us to explore whether we can determine if a given instance admits a TEF1 allocation for goods. Unfortunately, we show that this problem is NP-hard, with the following result.
\begin{theorem}
    Given an instance of the temporal fair division problem with goods and $n \geq 3$, determining whether there exists a \emph{TEF1} allocation is \emph{NP}-hard.
\end{theorem}

\begin{proof}
    We reduce from the \textsc{1-in-3-SAT} problem which is NP-hard.
    An instance of this problem consists of a conjunctive normal form formula $F$ with three literals per clause; it is a yes instance if there exists a truth assignment to the variables such that each clause has exactly one \texttt{True} literal, and a no instance otherwise.

    Consider an instance of \textsc{1-in-3-SAT} given by the CNF $F$ which contains $n$ variables $\{x_1,\dots,x_n\}$ and $m$ clauses $\{C_1,\dots, C_m\}$.
We construct an instance $\mathcal{I}$ with three agents and $2n+2$ goods.
    For each $i \in [n]$, we introduce two goods $t_i,f_i$.
    We also introduce two additional goods $s$ and $r$.
    Let the agents' (identical) valuations be defined as follows:
    \begin{equation*}
        v(g) = 
        \begin{cases} 
            5^{m+n-i} + \sum_{j \, : \, x_i \in C_j} 5^{m-j}, & \text{if } g = t_i, \\ 
            5^{m+n-i} + \sum_{j \, : \, \neg x_i \in C_j}  5^{m-j}, & \text{if } g = f_i, \\
            \sum_{j \in[m]} 5^{j-1}, & \text{if } g = r,\\
            \sum_{i \in [n]} 5^{m+i-1} + 2 \times \sum_{j \in [m]} 5^{j-1}, & \text{if } g = s.
        \end{cases}
    \end{equation*}
    Intuitively, for each variable index $i \in [n]$, we associate with it a unique value $5^{m+n-i}$.
    For each clause index $j \in [m]$, we also associate with it a unique value $5^{m-j}$.
    Note that no two indices (regardless of whether its a variable or clause index) share the same value.
    Then, the value of each good $t_i$ comprises of the unique value associated with $i$, and the sum over all unique values of clauses $C_j$ which $x_i$ appears as a \emph{positive literal} in; whereas the value of each good $f_i$ comprises of the unique value associated with $i$, and the sum over all unique values of clauses $C_j$ which $x_i$ appears as a \emph{negative literal} in.
    We will utilize this in our analysis later.

    Then, we have the set of goods $O = \{s,t_1,f_1,t_2,f_2,\dots,t_n,f_n,r\}$.
    Note that
    \begin{equation*}
        v(O) = v(s) + v(r) + \sum_{i \in [n]} v(t_i) + \sum_{i \in [n]} v(f_i.)
    \end{equation*}
    Also observe that
        $\sum_{i \in [n]} 5^{m+n-i} = \sum_{i \in [n]} 5^{m+i-1}$.
    Now, as each clause contains exactly three literals, we have
    \begin{equation*}
        \sum_{i \in [n]}\sum_{j : x_i \in C_j} 5^{m-j} + \sum_{i \in [n]}\sum_{j : \neg x_i \in C_j} 5^{m-j} = 3 \times \sum_{j \in [m]} 5^{j-1}.
    \end{equation*}
    Then, combining the equations above, we get that
    \begin{equation} \label{eqn:tefhardness_lem_vO_goods}
        v(O) = 3 \times \sum_{i \in [n]} 5^{m+i-1} + 6 \times \sum_{j \in [m]} 5^{j-1}.
    \end{equation}
        
    Let the goods be in the following order: 
    \begin{equation*}
        s,t_1,f_1,t_2,f_2,\dots,t_n,f_n,r.
    \end{equation*}
    We first prove the following result.

    \begin{lemma}\label{lem:tef1goodshard}
        There exists a truth assignment $\alpha$ such that each clause in $F$ has exactly one \texttt{True} literal if and only if there exists an allocation $\mathcal{A}$ such that $v(A_1) = v(A_2) = v(A_3)$ for instance $\mathcal{I}$.
    \end{lemma}
    \begin{proof}
        For the `if' direction, consider an allocation $\mathcal{A}$ such that $v(A_1) = v(A_2) = v(A_3)$.
        Then, we have that $O = A_1 \cup A_2 \cup A_3$ and $v(A_1) = v(A_2) = v(A_3) = \frac{1}{3} v(O)$.
        Since agents have identical valuations, without loss of generality, let $s \in A_1$.
        Then, since $v(A_1^1) = v(s) = \frac{1}{3} v(O)$, agent~$1$ should not receive any more goods after $s$, and each remaining good should go to agent~$2$ or $3$.

        Again, without loss of generality, we let $r \in A_2$.
        Then since $v(A_2) = \frac{1}{3}v(O)$, we have that
        \begin{align*}
            v(A_2 \setminus \{r\}) & = \left( \sum_{i \in [n]} 5^{m + i -1} + 2 \times \sum_{j \in [m]} 5^{j - 1} \right) - \sum_{j \in [m]} 5^{j - 1}\\
            & = \sum_{i \in [n]} 5^{m + i -1} + \sum_{j \in [m]} 5^{j - 1}.
        \end{align*} 
        Note that this is only possible if for each $i \in [m]$, $t_i$ and $f_i$ are allocated to different agents.
        The reason is because the only way agent~$1$ can obtain the first term of the above bundle value (less good $r$) is if he is allocated exactly one good from each of $\{t_i,f_i\}$ for all $i \in [n]$.

        Then, from the goods that exist in bundle $A_2$, we can construct an assignment $\alpha$:
        for each $i \in [n]$, let $x_i= \texttt{True}$ if $t_i \in A_2$ and $x_i = \texttt{False}$ if $f_i \in A_2$.
        Then, from the second term in the expression of $v(A_1 \setminus\{r\})$ above, we can observe that each clause has exactly one \texttt{True} literal (because the sum is only obtainable if exactly one literal appears in each clause, and our assignment will set each of these literals to \texttt{True}). 

        For the `only if' direction, consider a truth assignment $\alpha$ such that each clause in $F$ has exactly one \texttt{True} literal.
        Then, for each $i \in [n]$, let 
        \begin{equation*}
            \ell_i = 
            \begin{cases} 
                t_i & \text{if } x_i = \texttt{True} \text{ under } \alpha, \\
                f_i  & \text{if } x_i = \texttt{False} \text{ under } \alpha. 
            \end{cases}
        \end{equation*}
        We construct the allocation $\mathcal{A} = (A_1,A_2,A_3)$ where 
        \begin{equation*}
            A_1 = \{s\}, \quad A_2 = \{\ell_1, \dots, \ell_n, r\}, \quad \text{and} \quad A_3 = O \setminus (A_1 \cup A_2).
        \end{equation*}
        Again, observe that $\sum_{i \in [n]} 5^{m+n-i} = \sum_{i \in [n]} 5^{m+i-1}$.
        Also note that $v(A_1) = \frac{1}{3} v(O)$.
        Then, as each clause has exactly one \texttt{True} literal, 
            $v(A_2) = \sum_{i \in [n]}  \allowbreak 
 5^{m + i -1} + 2 \times \sum_{j \in [m]} 5^{j - 1}$,
        and together with (\ref{eqn:tefhardness_lem_vO_goods}), we get that 
        $v(A_3) = \frac{2}{3} v(O) -v(A_1) = v(A_1)$
        and hence $v(A_1) = v(A_2) = v(A_3)$,
        as desired.
    \end{proof}
    Now consider another instance $\mathcal{I}'$ that is similar to $\mathcal{I}$, but with an additional $21$ goods $\{g_1,\dots,g_{21}\}$.
    Let agents' valuations over these new goods be defined as follows:
    \begin{center}
    \begin{tabular}{c||c c c c c c c c c c c c c} 
         $\mathbf{v}$& $g_1$ & $g_2$ & $g_3$ & $g_4$ & $g_5$ & $g_6$ &  $g_7$ & $g_8$ & $g_9$ & $g_{10}$ & $g_{11}$\\ \hline \hline
         $1$ & $90$ & $80$ & $70$ & $100$ & $100$ & $100$ & $15$ & $10000$ & $11000$ & $12000$ & $20000$\\ 
         $2$ & $90$ & $70$ & $80$ & $100$ & $100$ & $100$ & $95$ & $10000$ & $11000$ & $12000$ & $20000$\\ 
         $3$ & $80$ & $90$ & $70$ & $100$ & $100$ & $100$ & $25$ & $10000$ & $11000$ & $12000$ & $20000$\\
         \hline
         & $g_{12}$ & $g_{13}$ & $g_{14}$ & $g_{15}$ & $g_{16}$ & $g_{17}$ & $g_{18}$ & $g_{19}$ & $g_{20}$ & $g_{21}$\\ \hline \hline
         $1$  & $20000$ & $20000$ & $20000$ & $20000$ & $20000$ & $20000$ & $20000$ & $20000$ & $19010$ & $18005$\\ 
         $2$ & $20000$ & $20000$ & $20000$ & $20000$ & $20000$ & $20000$ & $12000$ & $12000$ & $19085$ & $14106$\\ 
         $3$ & $20000$ & $18500$ & $20000$ & $20000$ & $20000$ & $20000$ & $20000$ & $20000$ & $19010$ & $19496$\\
    \end{tabular}
    \end{center}
    Then, we have the set of goods $O' = O \cup \{g_1,\dots,g_{21}\}$.

    Let the goods be in the following order:
    \begin{equation*}
        s, t_1,f_1,t_2,f_2,\dots,t_n,f_n,r, g_1,\dots,g_{21}.
    \end{equation*}
    We now present the final lemma that will give us our result.
    \begin{lemma} \label{lemma-code}
        If there exists a partial allocation $\mathcal{A}^{2n+2}$ over the first $2n+2$ goods such that $v(A^{2n+2}_1) = v(A^{2n+2}_2)$, then there exists a \emph{TEF1} allocation $\mathcal{A}$.
        Conversely, if there does not exist a partial allocation $\mathcal{A}^{2n+2}$ over the first $2n+2$ goods such that $v(A^{2n+2}_1) = v(A^{2n+2}_2)$, then there does not exists a \emph{TEF1} allocation $\mathcal{A}$.
    \end{lemma}
    We use a program as a gadget to verify the lemma (the code can be found in Appendix \ref{app:code}), leveraging its output to support its correctness.
    Specifically, if there exists a partial allocation $\mathcal{A}^{2n+2}$ over the first $2n+2$ goods such that $v(A^{2n+2}_1) = v(A^{2n+2}_2)$, then our program will show the existence of a TEF1 allocation by returning all such TEF1 allocations.
    If there does not exist such a partial allocation, our program essentially does an exhaustive search to show that a TEF1 allocation does not exist.
    This lemma shows that there exists a TEF1 allocation over $O'$ if and only if $v(A^{2n+2}_1) \neq v(A^{2n+2}_2)$, and by Claim~\ref{lem:tef1goodshard}, this implies that a TEF1 allocation over $O'$ exists if and only if there is a truth assignment $\alpha$ such that each clause in $F$ has exactly one \texttt{True} literal.
\end{proof}
However, we note that the above approach cannot be extended to show hardness for the setting with chores.
Nevertheless, we are able to show a similar, though weaker, intractability result for the case of chores in general.
The key difference is that we assume that we can start from any partial TEF1 allocation instead of the empty allocation.
\begin{theorem}
    For every $t\in [T]$, given any partial \emph{TEF1} allocation $\mathcal{A}^t$ for chores, deciding if there exists an allocation $\mathcal{A}$ that is \emph{TEF1} is \emph{NP}-hard.
\end{theorem}
\begin{proof}
    We reduce from the NP-hard problem \textsc{Partition}.
    An instance of this problem consists of a multiset $S$ of positive integers; it is a yes-instance if $S$ can be partitioned into two subsets $S_1$ and $S_2$ such that the sum of the numbers in $S_1$ equals the sum of the numbers in $S_2$, and a no-instance otherwise.

    Consider an instance of \textsc{Partition} given by a multiset set $S = \{s_1,\dots, s_m\}$ of $m$ positive integers. Then, we construct a set $S' = \{s'_1,\dots s'_m\}$ such that for each $j \in [m]$, $s'_j = s_m - K$ where $K := \max\{s_1,\dots,s_m\} + \varepsilon$ for some small $\varepsilon > 0$.
    We then scale members of $S'$ such that they sum to $-2$, i.e., $\sum_{s' \in S'} s' = -2$.

    Next, we construct an instance with four agents and $m+4$ chores $O = \{b_1,b_2,b_3,b_4,c_1,\dots,c_m\}$, where agents have the following valuation profile $\mathbf{v}$ for $j \in \{1,\dots,m\}$:
    \begin{center}
    \begin{tabular}{c||c c c c c c c c c c} 
         $\mathbf{v}$& $b_1$ & $b_2$ &  $b_3$ & $b_4$ & $c_1$ & $\dots$ &  $c_j$ & $\dots$ & $c_m$ \\ \hline \hline
         $1$ & \circled{$-1$}&  $0$&  $0$& $0$ & $-1$ & $\dots$ &  $-1$ &  $\dots$ & $-1$\\  
         $2$&  $-1$& \circled{$-1$}& $-1$& $-1$ & $s'_1$ & $\dots$& $s'_j$ & $\dots$ & $s'_m$\\ 
         $3$&  $-1$&  $-1$& \circled{$-1$}& $-1$ & $s'_1$ &  $\dots$& $s'_j$ & $\dots$ & $s'_m$\\ 
         $4$& $0$ & $0$ & $0$  & \circled{$-1$} & $-1$ & $\dots$ & $-1$ & $\dots$ & $-1$\\
    \end{tabular}
    \end{center}
    Also, suppose we are given the partial allocation $\mathcal{A}^4$ where for each $i \in \{1,2,3,4\}$, chore $b_i$ is allocated to agent $i$, as illustrated in the table above.
    Note that the partial allocation $\mathcal{A}^4$ is TEF1.
    
    We first establish the following two lemmas.
    The first lemma states that after chores $b_1,b_2,b_3,b_4$ are allocated, in order to maintain TEF1, each remaining chore in $\{c_1,\dots,c_m\}$ cannot be allocated to either agent $1$ or agent $4$. The result is as follows.
    \begin{lemma} \label{lem:chores_ef1_hard_agents14}
        In any \emph{TEF1} allocation, agents~$1$ and $4$ cannot be allocated any chore in $\{c_1,\dots,c_m\}$.
    \end{lemma}
    \begin{proof}
        Consider any TEF1 allocation $\mathcal{A}$.
        Suppose for a contradiction that at least one of agent~$1$ and $4$ is allocated a chore in $\{c_1,\dots,c_m\}$.
        Assume without loss of generality that agent~$1$ was the first (if not only) agent that received such a chore.
        
        Consider the first round $j+4$ (for some $j \in [m]$) whereby agent~$1$ is allocated some chore $c_j \in \{c_1,\dots,c_m\}$.
        Then,
        \begin{equation*}
            v_1(A_1^{j+4} \setminus \{b_1\}) = -1 < 0 = v_1(A_4^{j+4}), 
        \end{equation*}
        a contradiction to $\mathcal{A}$ being TEF1.
    \end{proof}
    The second lemma states that in any TEF1 allocation, the sum of values that agents $2$ and $3$ obtain from the chores in $\{c_1,\dots,c_m\}$ that are allocated to them must be equal. We formalize it as follows.
    \begin{lemma}\label{lem:chores_ef1_hard_agents23}
        In any \emph{TEF1} allocation, let $C_2, C_3$ be the subsets of $\{c_1, \allowbreak \dots, c_m\}$ that were allocated to agents $2$ and $3$ respectively.
        Then, $v_2(C_2) = v_3(C_3)$.
    \end{lemma}
    \begin{proof}
        Consider any TEF1 allocation $\mathcal{A}$.
        Suppose for a contradiction that $v_2(C_2) \neq v_3(C_3)$.
        Since $v_2(C_2) + v_3(C_3) = \sum_{s' \in S'} s' = -2$, it means one of $\{v_2(C_2),v_3(C_3)\}$ is strictly less than $-1$, and the other is strictly more than $-1$.
        Without loss of generality, assume $v_2(C_2) > v_3(C_3)$, i.e., $v_3(C_3) < -1$.
        We get that
        \begin{equation*}
            v_3(A_3 \setminus \{b_3\}) = v_3(C_3) < -1 = v_3(A_1),
        \end{equation*}
        contradicting the fact that $\mathcal{A}$ is a TEF1 allocation.
    \end{proof}
    We will now prove that there exists an allocation $\mathcal{A}$ satisfying TEF1 if and only if the set $S$ can be partitioned into two subsets of equal sum.

    For the `if' direction, suppose $S = \{s_1,\dots,s_m\}$ can be partitioned into two subsets $S_1,S_2$ of equal sum.
    This means that $S' = \{s'_1,\dots,s'_m\}$ can be correspondingly partitioned into two subsets $S'_1,S'_2$ of equal sum (of $-1$ each).
    Let $C_1,C_2$ be the partition of chores in $\{c_1,\dots,c_m\}$ with values corresponding to the partitions $S'_1,S'_2$ respectively.
    Then we allocate all chores in $C_1$ to agent~$2$ and all chores in $C_2$ to agent~$3$.
    By \Cref{lem:chores_ef1_hard_agents14}, we have that agents~$1$ and $4$ cannot envy any other agent at any round.
    Also, for any round $t \in [T]$ and $i,j \in \{2,3\}$ where $i \neq j$,
        $v_i(A_i^t \setminus \{b_i\}) \geq -1 \geq v_i(A_j^t)$,
    and for all $i \in \{2,3\}$ and $k \in \{1,4\}$,
        $v_i(A^t_i \setminus \{b_i\}) \geq -1 = v_i(A_k^t)$.
    Thus, the allocation $\mathcal{A}$ that, for each $i \in \{1,2,3,4\}$, allocates $b_i$ to agent $i$ and for each $j \in \{2,3\}$, allocates $C_j$ to agent $j$, is TEF1.
    
    For the `only if' direction, suppose we have an allocation $\mathcal{A}$ satisfying TEF1.
    By \Cref{lem:chores_ef1_hard_agents14}, it must be that any chore in $\{c_1,\dots,c_m\}$ is allocated to either agent $2$ or $3$.
    Let $C_2,C_3$ be the subsets of chores in $\{c_1,\dots,c_m\}$ that are allocated to agents~$2$ and $3$ respectively, under $\mathcal{A}$.
    Then, by Lemma \ref{lem:chores_ef1_hard_agents23}, we have that $v_2(C_2) = v_3(C_3)$.
    By replacing the chores with their corresponding values, we get a partition of $S'$ into two subsets of equal sums, which in turn gives us a partition of $S$ into two subsets of equal sum.   
\end{proof}

\section{Compatibility of TEF1 and Efficiency} \label{sec:efficiency}
In traditional fair division, many papers have focused on the existence and computation of fair and efficient allocations for goods or chores, with a particular emphasis on simultaneously achieving EF1 and \emph{Pareto-optimality (PO)} \citep{barman2018fairandefficient,CaragiannisKuMo19}. 
In this section, we explore the compatibility between TEF1 and PO. We begin by defining PO as follows.

\begin{definition}[Pareto-optimality] \label{def:po}
We say that an allocation $\mathcal{A}$ is \emph{Pareto-optimal (PO)} if there does not exist another allocation $\mathcal{A}'$ such that for all $i\in N$, $v_i(A'_i)\geq v_i(A_i)$, and for some $j\in N$, $v_j(A'_j)> v_j(A_j)$. If such an allocation $\mathcal{A}'$ exists, we say that $\mathcal{A}'$ \emph{Pareto-dominates} $\mathcal{A}$.
\end{definition}
Observe that for any $\mathcal{A}$ that is PO, any partial allocation $\mathcal{A}^t$ for $t \leq [T]$ is necessarily PO as well.
We demonstrate that PO is incompatible with TEF1 in this setting, even under very strong assumptions (of two agents and two types of items), as illustrated by the following result.

\begin{proposition} \label{prop_tef1_po}
    For any $n \geq 2$, a \emph{TEF1} and \emph{PO} allocation for goods or chores may not exist, even when there are two types of items.
\end{proposition}
\begin{proof}
    We first prove the result for goods. Consider an instance with two agents and four goods $O = \{g_1,g_2,g_3,g_4\}$, with the following valuation profile:
    \begin{center}
        \begin{tabular}{l||cccc}
            $\mathbf{v}$ &$g_1$ & $g_2$ & $g_3$ & $g_4$\\ \hline\hline
            1&$1.1$&  $1.1$&  $2$&  $2$\\
            2&$2$&  $2$&  $1.1$&  $1.1$\\
        \end{tabular}
    \end{center}
    Observe that the first two goods must be allocated to different agents, otherwise TEF1 will be violated after the second good is allocated. 
    Without loss of generality, suppose that agent $1$ receives $g_1$ and agent $2$ receives $g_2$.
    We have $v_1(g_1)<v_1(\{g_2,g_3,g_4\})-v_1(g_3)$ and $v_2(g_2)<v_2(\{g_1,g_3,g_4\})-v_2(g_1)$, thereby showing that EF1 will be violated if $g_3$ and $g_4$ are allocated to the same agent.
    
    Thus, in any TEF1 allocation $\mathcal{A}$, agent $1$ must receive one good from $\{g_1,g_2\}$ and one good from $\{g_3,g_4\}$. 
    However, observe that every such allocation $\mathcal{A}$ is Pareto-dominated by the allocation where agent $2$ receives bundle $\{g_1,g_2\}$ and agent $1$ receives bundle $\{g_3,g_4\}$. 
    This proof can be extended to the case of $n \geq 3$ simply by adding dummy agents who have zero value for each good, and observing that they cannot receive any item in a PO allocation.
    As such, a TEF1 and PO allocation cannot be guaranteed to exist, even when when there are two types of chores.

    Next, we prove the result for chores.
    Consider an instance with $n \geq 2$ agents and $2n$ chores $O = \{c_1,\dots,c_{2n} \}$, with the following valuation profile:
    \begin{table}[H]
        \centering
        \begin{tabular}{l||cccccc}
            $\mathbf{v}$ &$c_1$ & $\dots$ & $c_n$ & $c_{n+1}$ & $\dots$ & $c_{2n}$\\ \hline\hline
            $1$&$-1.1$&  $\dots$ & $-1.1$ & $-2$ & $\dots$ & $-2$\\
            $2$&$-2$&  $\dots$& $-2$ & $-1.1$ &$\dots$&  $-1.1$\\
            $3$&$-2$&  $\dots$& $-2$&$-2$&  $\dots$& $-2$\\
            $\vdots$ & $\vdots$ & & $\vdots$& $\vdots$ & & $\vdots$\\
            $n$&$-2$&  $\dots$& $-2$&$-2$&  $\dots$& $-2$
        \end{tabular}
    \end{table}
    In this instance, agent~$1$ has value $-1.1$ for each of the first $n$ chores, and value $-2$ for the last $n$ chores. 
    Agent~$2$ has value $-2$ for the first $n$ chores, and value $-1.1$ for the last $n$ chores.
    If $n \geq 3$, then agents $3,\dots,n$ have value $-2$ for all chores. 
    
    Observe that each agent must receive one of the first $n$ chores to avoid violating TEF1 within the first $n$ rounds. 
    We now show that each agent must also receive one of the final $n$ chores, otherwise TEF1 will be violated. 
    Suppose for contradiction that in the final allocation $\mathcal{A}$, some agent $i \in N$ is allocated at least two chores from $\{c_{n+1},\dots,c_{2n}\}$. Then for each $i \in N$, let $c'_i := \argmin_{c \in A_i} v_i(c)$. We get that
    \begin{equation}
    v_i(A_i \setminus \{c'_i\})\leq
        \begin{cases}
            -5.1 + 2 = -3.1 & \text{if }i=1,\\
            -4.2 + 2 = -2.2& \text{if }i=2,\\
            -6 + 2 = -4 & \text{if }i\in \{3,\dots,n\}.
        \end{cases}
    \end{equation}
    By the pigeonhole principle, there exists some other $j \in N \setminus\{i\}$ that receives no chore from $\{c_{n+1},\dots,c_{2n}\}$, giving us
    \begin{equation}
    v_i(A_j)=
        \begin{cases}
            -1.1 & \text{if }i=1,\\
            -2 & \text{if }i=2,\\
            -2 & \text{if }i\in \{3,\dots,n\}.
        \end{cases}
    \end{equation}
    Consequently, agent $i$ would envy agent $j$ even after removing one chore from her own bundle, and TEF1 is violated.
    Thus, in any TEF1 allocation, each agent must receive exactly one chore from $\{c_1,\dots,c_n\}$ and exactly one chore out of $\{c_{n+1},\dots,c_{2n}\}$. 
    
    However, any such allocation is Pareto-dominated by another allocation where agent $1$ receives exactly two chores from $\{c_{1},\dots,c_{n}\}$ and no chores from $\{c_{n+1},\dots,c_{2n}\}$, and agent $2$ receives no chores from $\{c_{1},\dots,c_{n}\}$ and exactly two chores from $\{c_{n+1},\dots,c_{2n}\}$.
    As such, a TEF1 and PO allocation cannot be guaranteed to exist, even when there are two types of chores.
\end{proof}
Despite this non-existence result, one may still wish to obtain a TEF1 and PO outcome when the instance admits one. However, the following results show that this is not computationally tractable.
\begin{theorem} \label{thm:tef1_po_nphard_goods}
    Determining whether there exists a \emph{TEF1} allocation that is \emph{PO} for goods is \emph{NP}-hard, even when $n = 2$.
\end{theorem}

\begin{theorem} \label{thm:tef1_po_nphard_chores}
    Determining whether there exists a \emph{TEF1} allocation that is \emph{PO} for chores is \emph{NP}-hard, even when $n=2$.
\end{theorem}

The above results essentially imply that even determining whether an instance admits a TEF1 and \emph{utilitarian-maximizing} (i.e., sum of agents' utilities) allocation is computationally intractable, since a utilitarian-welfare maximizing allocation is necessarily PO.
In fact, for the case of goods, we can make a stronger statement relating to the general class of \emph{p-mean welfares}, defined as follows.\footnote{Note that we cannot say the same for chores as when agents' valuations are negative, as the $p$-mean welfare may be ill-defined. Moreover, the fairness guarantees in the welfare measures that do work are not well-established, apart from concepts like utilitarian or egalitarian welfare.}
\begin{definition} \label{defn-pmean} 
    Given $p \in (-\infty, 1]$ and an allocation  $\mathcal{A} = (A_1,\dots, \allowbreak A_n)$ of goods, the
    \emph{$p$-mean welfare} 
    is 
        $\left( \frac{1}{n} \sum_{i \in N} v_i(A_i)^p \right)^{1/p}$.
\end{definition}
In the context of fair division, $p$-mean welfare has been traditionally and well-studied for the setting with goods \citep{barman2020pmean,chaudhury2021fairefficient}, although it has recently been explored for chores as well \citep{Eckart2024}.
Importantly, $p$-means welfare captures a spectrum of commonly studied fairness objectives in fair division.
For instance, setting $p = 1$ (resp. $p= -\infty$) would correspond to the utilitarian (resp. egalitarian) welfare. Setting $p \to 0$ corresponds to maximizing the geometric mean, which is also known as the Nash welfare \citep{CaragiannisKuMo19}. 

Then, from our construction in the proof of Theorem~\ref{thm:tef1_po_nphard_goods} (for goods), we have that an allocation is TEF1 and PO if and only if it also maximizes the $p$-mean welfare, for all $p \in (-\infty,1]$, thereby giving us the following corollary.

\begin{corollary}
    For all $p \in (-\infty,1]$, determining whether there exists a \emph{TEF1} allocation that maximizes $p$-means welfare is \emph{NP}-hard, even when $n = 2$.
\end{corollary}

\section{Multiple Items per Round} \label{sec:multigoods}
We now move to the setting where multiple items may arrive at each round. 
Recall that Lemma~\ref{lem:transform} states that instances where multiple items arrive per round can be transformed into instances where a single item arrives per round. 
Accordingly, in this section, we focus on results for restricted settings which are incompatible with Lemma~\ref{lem:transform}.
We begin by showing that when there are two timesteps, a TEF1 allocation can be computed efficiently.
\begin{theorem} \label{thm:multi_T=2}
    When $T = 2$, a \emph{TEF1} allocation for goods or chores exists and can be computed in polynomial time.
\end{theorem}
\begin{proof}
    Let $\mathcal{A}^1$ and $\mathcal{B} = \mathcal{A}^2\setminus\mathcal{A}^1$ be the allocations of item sets $O_1$ and $O_2$ respectively. 
    Note that while we are in the setting whereby $O_1 = O_2$, we can simply relabel item.
    
    We first address a special case. When allocating chores, for each $t\in \{1,2\}$ such that $|O_t|<n$ (i.e., there are less chores than agents in either round), add $n-|O_t|$ zero-valued dummy chores to $O_t$.
    
    To obtain $\mathcal{A}^1$, we allocate the items in the first round in a round-robin fashion, with picking sequence $(1,\dots,n)^*$. That is, agent $1$ picks their most preferred item, followed by agent $2$, and so on until agent $n$, after which the sequence restarts. The items arriving in the second round are also allocated in a round-robin fashion to obtain $\mathcal{B}$, but with picking sequence $(n,\dots,1)^*$. The round-robin algorithm is well-known to satisfy EF1 for both the goods and chores settings \citep{AzizCaIg22chores}, so we know that $\mathcal{A}^1$ and $\mathcal{B}$ are EF1. It remains to show that $\mathcal{A}^2=\mathcal{A}^1 \cup \mathcal{B}$ is EF1.

    Consider an arbitrary pair of agents $i,j$. If $i<j$, then 
    $$v_i(A^1_i) \geq v_i(A^1_j),$$
    because $i$ precedes $j$ in the picking sequence for allocation $\mathcal{A}$. Similarly, if $i>j$, then $$v_i(B_i) \geq v_i(B_j).$$ Note that these inequalities hold for both goods and chores.
    
    We now prove our result for goods. Consider an arbitrary agent $i$. Since $\mathcal{A}^1$ and $\mathcal{B}$ are EF1, we know that for any agent $j\neq i$, there exists a good $g_a\in A^1_j$ such that $v_i(A^1_i) \geq v_i(A^1_j\setminus \{g_a\})$, and there exists a good $g_b\in B_j$ such that $v_i(B_i) \geq v_i(B_j\setminus \{g_b\})$. Therefore for any agent $j<i$, there exists $g_a\in A^1_j$ such that
    \begin{equation*}
        v_i(A_i^2)
        =v_i(A^1_i \cup B_i)
        \geq v_i(A^1_j)-v_i(g_a) + v_i(B_j)
        = v_i(A_j^2)-v_i(g_a)=v_i(A_j^2\setminus \{g_a\}).
    \end{equation*}
    Similarly, for any $j>i$, there exists $g_b\in B_j$ such that $v_i(A_i^2)\geq v_i(A_j^2\setminus \{g_b\})$.

    We next prove our result for chores. Again consider an arbitrary agent $i$. Due to $\mathcal{A}^1$ and $\mathcal{B}$ satisfying EF1, for any agent $j\neq i$, there exists a chore $c_a\in A^1_i$ such that $v_i(A^1_i\setminus \{c_a\}) \geq v_i(A^1_j)$, and there exists a chore $c_b\in B_j$ such that $v_i(B_i\setminus \{c_b\}) \geq v_i(B_j)$. Therefore for any $j<i$, there exists $c_a\in A^1_i$ such that
    \begin{equation*}
        v_i(A_i^2\setminus \{c_a\})
        =v_i(A^1_i\setminus \{c_a\})+v_i(B_i)
        \geq v_i(A^1_j) + v_i(B_j)
        = v_i(A_j^2).
    \end{equation*}
    Similarly, for any $j>i$, there exists $c_b\in B_i$ such that $v_i(A_i^2\setminus \{c_b\})\geq v_i(A_j^2)$. This concludes the proof.
\end{proof}
For the remainder of Section~\ref{sec:multigoods}, we consider the \emph{repeated} setting (as similarly studied by \citet{igarashi2023repeatedfairallocation} and \citet{caragiannis2024repeatedmatching}), whereby the same set of items appears at each round.
In our setting of a single item per round, this reduces to the same item appearing at every round, and a simple round-robin algorithm over the agents will suffice in achieving TEF1 at every round.
However, a positive result in the single-item case does not translate to a positive result in the multi-item case here, since we are imposing a constraint on the set of items that can appear at each round.

It remains an open question whether a TEF1 allocation exists in general for this setting.
However, we can show that, perhaps surprisingly, determining whether repeating the \emph{same allocation} over the items at \emph{every round} results in a TEF1 allocation is NP-hard.

We first define the concept of a \emph{repetitive allocation} $\mathcal{A}$, whereby $\mathcal{A}^t \setminus \mathcal{A}^{t-1} = \mathcal{A}^{t'} \setminus \mathcal{A}^{t'-1}$ for each $t,t' \in [T]$. 
Then, for an initial allocation $\mathcal{A}^0 = (\varnothing, \dots, \varnothing)$, we have the following result.

\begin{theorem} \label{thm:multi_hardness}
    Determining whether there exists a repetitive allocation $\mathcal{A} = (A_1,\dots,A_n)$ which is \emph{TEF1} for goods or chores is \emph{NP}-hard, even when $T=2$ and agents have identical valuations.
\end{theorem}
\begin{proof}
    For both cases of goods and chores,
    we reduce from NP-hard problem \textsc{Multiway Number Partitioning} \citep{Graham1969}.
    An instance of this problem consists of a positive integer $\kappa$ and a multiset $S =\{s_1\dots s_m\}$ of $m$ non-negative integers whose sum is $\kappa W$; it is a yes-instance if $S$ can be partitioned into $\kappa$ subsets such that the sum of integers in each subset is $W$, and a no-instance otherwise.

    Consider an instance of \textsc{Multiway Number Partitioning} given by a positive integer $\kappa$ and a multiset $S = \{s_1,\dots,s_m\}$ of $m$ nonnegative integers.

    We first prove the result for goods.
    We construct an instance with $\kappa+1$ agents and $m+1$ goods in each round:
    $O_1 = \{g_1,\dots,g_{m+1}\}$ and $O_2 = \{g'_1,\dots,g'_{m+1} \}$, where agents have an identical valuation function $v$ defined as follows:
    \begin{equation*}
        v(g_j) = v(g'_j) = 
        \begin{cases}
        s_j, & \text{if } j \leq m, \\
      2W, & \text{if } j = m+1.         \end{cases}
    \end{equation*}

    We will now prove that there exists a repetitive TEF1 allocation $\mathcal{A}$ if and only if the set $S$ can be partitioned into $\kappa$ subsets with equal sums (of $W$ each).

    For the `if' direction, consider a $\kappa$-way partition $\mathcal{P} = \{P_1, \dots P_\kappa\}$ of $S$ with equal sums. 
    We construct allocations $\mathcal{A}^1$ and $\mathcal{A}^2$ such that the goods in both rounds are allocated identically, and show that $\mathcal{A}^2$ satisfies TEF1.
    
    For each $i \in \{1, \dots, \kappa\}$, allocate the goods corresponding to the elements of subset $P_i$ to agent $i$, and the good $g_{m+1} $ to agent $\kappa+1$. 
    Then, in $\mathcal{A}^1$, for each agent $i\in [\kappa]$, $v(A^1_i) = \sum_{g \in P_i}g = W$, and $v(A^1_{\kappa+1}) = v(\{g_{m+1}\}) = 2W$.
    It is easy to verify that $\mathcal{A}^1$ is TEF1: no agent $i \in [\kappa]$ will envy another agent $j \in [\kappa] \setminus \{i\}$, as they have the same bundle value, and agent $i$'s envy towards agent $\kappa+1$ can be removed by simply dropping $g_{m+1} \in A^1_{\kappa+1}$.
    Also, agent $\kappa+1$, having a strictly higher bundle value of $2W$, will not envy any agent $i \in [\kappa]$, whose bundle value is $W$.
    
    Next, we consider $\mathcal{A}^2$. Recall that $\mathcal{A}=\mathcal{A}^2$ is a repetitive allocation, and thus the items in the second round are allocated identically to the items in the first round. For each agent  $i\in [\kappa]$, $v(A^2_i) = 2W$, and  $v(A^2_{\kappa+1}) = 4W$.
    It is also easy to verify that $\mathcal{A}^2$ is TEF1: each pair of agents $i,j \in [\kappa]$ have the same bundle value, and the envy that any agent $i\in [\kappa]$ has towards agent $\kappa+1$ can be removed by dropping good $g_{m+1} \in A^2_{\kappa+1}$.
    Finally, agent $\kappa+1$ has a strictly higher bundle value than any agent $i \in [\kappa]$.

    For the `only if' direction, suppose we have a repetitive allocation $\mathcal{A}^2$ which satisfies TEF1.
    Since agents have identical valuation functions, without loss of generality, suppose that agent $\kappa+1$ receives good $g_{m+1}$ under $\mathcal{A}^1$.
    Then, $v(A^2_{\kappa+1} \setminus \{g_{m+1}\}) \geq 2W$.
    In order for $\mathcal{A}^2$ to be TEF1, we must have that $v(A^2_i) \geq 2W$ for each $i \in [\kappa]$ (so that an agent $i \in [\kappa]$ does not envy agent $\kappa+1$).
    This means that for each $i \in [\kappa]$, $v(A^1_i) \geq W$, but since $\sum_{j\in [m]}s_j=\kappa W$, this is only possible if there is a $\kappa$-way partition of $S$ such that each subset has a sum of $W$.

     We now prove the result for the case of chores.
    We construct a set $S' = \{s'_1,\dots, s'_m\}$ such that for each $j \in [m]$, $s'_j = -K + s_j$ where $K := \max\{s_1,\dots,s_m\}$.
    Observe that $S'$ contains non-positive integers.
    Let $W' := \frac{1}{\kappa}\sum_{j\in [m]}s'_j$
    
    Then, we construct an instance with $\kappa+1$ agents and $m+1$ chores in each round:
    $O_1 = \{c_1,\dots,c_{m+1} \}$ and $O_2 = \{c'_1,\dots,c'_{m+1} \}$, where agents have an identical valuation function $v$ defined as follows:
    \begin{equation*}
        v(c_j) = v(c'_j) = 
        \begin{cases}
        s'_j, & \text{if } j \leq m, \\
      2W', & \text{if } j = m+1.         \end{cases}
    \end{equation*}
    We will now prove that there exists a repetitive TEF1 allocation $\mathcal{A}$ if and only if the set $S$ can be partitioned into $\kappa$ subsets with equal sums (of $W$ each).

    For the `if' direction, consider a $\kappa$-way partition $\mathcal{P} = \{P_1,\dots,P_\kappa\}$ of $S$ with equal sums (of $W$ each).
    This means that $S'$ can also be partitioned into $\kappa$ subsets of equal sums (with the same partition $\mathcal{P}$; let the sum be $W'$).
    We construct allocations $\mathcal{A}^1$ and $\mathcal{A}^2$ such that the chores in both rounds are allocated identically, and show that $\mathcal{A}^2$ satisfies TEF1.

    For each $i \in \{1,\dots,\kappa\}$, allocate the chores corresponding to the elements of subset $P_i$ to agent $i$, and the chore $c_{m+1}$ to agent $\kappa+1$.
    Then, in $\mathcal{A}^1$, for each agent $i \in [\kappa]$, $v(A^1_i) = \sum_{c \in P_i} c = W'$, and $v(A^1_{\kappa+1}) = v(\{c_{m+1}\}) = 2W'$.
    It is easy to verify that $\mathcal{A}^1$ is TEF1: every pair of agents $i,j \in [\kappa]$ has the same bundle value, and each agent $i\in [\kappa]$ has a higher bundle value than agent $\kappa+1$.
    Also, agent $\kappa+1$ will not envy any agent $i \in [\kappa]$ after removing chore $c_{m+1} \in A^1_{\kappa+1}$.

    Next, we consider $\mathcal{A}^2$.
    For each agent $i \in [\kappa]$, $v(A^2_i) = 2W'$, and $v(A^2_{\kappa+1}) = 4W'$.
    We verify that $\mathcal{A}^2$ is TEF1: again, each pair of agents $i,j \in [\kappa]$ has the same bundle value, and each agent $i\in [\kappa]$ has a higher bundle value than agent $\kappa+1$.
    Also, agent $\kappa+1$ will not envy any agent $i \in [\kappa]$ after removing chore $c_{m+1} \in A^2_{\kappa+1}$. 

    For the `only if' direction, suppose we have a repetitive allocation $\mathcal{A}^2$ which satisfies TEF1.
    Since agents have identical valuation functions, without loss of generality, suppose that agent $\kappa+1$ receives chore $c_{m+1}$ under $\mathcal{A}^1$.
    Then, $v(A^2_{\kappa+1} \setminus \{c_{m+1}\}) \leq 2W'$.
    In order for $\mathcal{A}^2$ to be TEF1, we must have that $v(A^2_i) \leq 2W'$ for each $i \in [\kappa]$ (so that agent $\kappa+1$ will not envy any agent $i \in [\kappa]$).
    This means that for each $i \in [\kappa]$, $v(A^1_i) \leq W'$, but since $\sum_{j\in [m]}s_j'=\kappa W'$, this is only possible if there is a $\kappa$-way partition of $S'$ such that each subset has a sum of $W'$ (i.e. there is a $\kappa$-way partition of $S$ such that each subset has a sum of $W$).
\end{proof}

\section{Conclusion}
In this work, we studied the informed online fair division of indivisible items, with the goal of achieving TEF1 allocations. For both goods and chores, we demonstrated the existence of TEF1 allocations in four special cases and provided polynomial-time algorithms for each case. Additionally, we showed that determining whether a TEF1 allocation exists for goods is NP-hard, and presented a similar, though slightly weaker, intractability result for chores. We further established the incompatibility between TEF1 and PO, which extends to an incompatibility with $p$-mean welfare. Finally, we explored the special case of multiple items arriving at each round.

Numerous potential directions remain for future work, including revisiting variants of the standard fair division model.
Examples include studying the existence (and polynomial-time computability) of allocations satisfying a temporal variant of the weaker \emph{proportionality up to one item} property (as defined by \citet{CFS17}), which would be implied by EF1; studying group fairness \citep{aleksandrov2018groupef,aziz2020groupefchores,benabbou2019groupfairness,conitzer2019groupfairness,kyropoulou2019groupallocation,Scarlett2023}; considering the more general class of \emph{submodular} valuations \citep{ghodsi2022submodular,montanari2024submodular,SuksompongTe23,uziahu2023submodular}; examining the \emph{house allocation} model where each agent gets a single item \citep{Choo2024,gan2019envy}, which was partially explored by \citet{micheel2024repeatedhouse}, or even looking at more general settings with additional size constraints \citep{BarmanKhShSrAAAI2023,BarmanKhShSr2023,Elkind2024}.
It would also be interesting to extend our results, which hold for the cases of goods and chores separately, to the more general case of mixed manna, in which items can be simultaneously viewed as goods by some agents and as chores by others (see, e.g., \cite{AzizCaIg22chores}). In fact, with an appropriate modification of the instance, we can extend Theorem~\ref{thm:2agents} to show that a TEF1 allocation exists in the mixed manna setting when there are two agents, which we detail in Appendix~\ref{app:mixed_manna}.

\bibliographystyle{plainnat}
\bibliography{abb,sample}

\appendix

\section{Counterexample for Goods when $n \geq 3$} \label{app:counterexample}
As mentioned in the main text,
\citet[Thm. 4.2]{he2019fairerfuturepast} used the following counterexample to show that a TEF1 allocation may not exist for goods when $n=3$. Note that in this counterexample, one good arrives at each round.
\begin{table}[H]
    \centering
    \begin{tabular}{l||cccccccccccccccccccc}
        $\mathbf{v}$ &$g_1$ & $g_2$ & $g_3$ & $g_{4-6}$ & $g_7$ & $g_8$ & $g_9$ & $g_{10}$ \\ \hline\hline
        1& $0.9$& $0.8$& $0.7$ &$1$ & $0.15$ & $100$ & $110$ & $120$\\
        2& $0.9$& $0.7$& $0.8$&$1$ & $0.95$ & $100$ & $110$ & $120$ \\
        3& $0.8$& $0.9$& $0.7$&$1$ & $0.25$ & $100$ & $110$ & $120$\\
        \hline 
         & $g_{11-12}$&$g_{13}$ & $g_{14-17}$ & $g_{18-19}$ & $g_{20}$ & $g_{21-22}$ & $g_{23}$\\
        \hline\hline 
        $1$ & $200$  & $200$ & $200$ & $200$& $200$& $200$& $200$ & \\
        $2$ & $200$ & $200$ & $200$ & $120$& $200$& $120$& $200$ & \\
        $3$ & $200$ & $185$ & $200$ & $200$& $200$& $200$& $200$ &
    \end{tabular}
\end{table}
For completeness, we briefly explain the counterexample. There are three parts to this example, which ultimately ensure that in any TEF1 allocation, after $g_{22}$ is allocated, agent $2$ envies both other agents, and one other agent envies agent $2$. As a result, $g_{23}$ cannot be allocated to any agent without violating TEF1.

The first part consists of goods $g_1$ to $g_7$. The instance is constructed such that after all of the goods in this part have been allocated in a TEF1 manner, the possible envy relations are restricted. Specifically, we have that after round $7$, agent $3$ cannot envy agent $1$, and that agent $2$ cannot envy agents $1$ or $3$.

The second part consists of goods $g_8$ to $g_{16}$, and builds on top of the previous envy restriction to ensure that after round $16$, agent $2$ is envied by either agent $1$ or agent $3$ in any TEF1 allocation.

The final part consists of goods $g_{17}$ to $g_{23}$. Since agent $2$ is envied by some agent at the start, it cannot receive good $g_{17}$ or $g_{20}$, and must receive one of $\{g_{18},g_{19}\}$ and one of $\{g_{21},g_{22}\}$. This causes agent $2$ to envy both other agents, while one of the other agents continues to envy agent $2$. Therefore, TEF1 will be violated regardless of which agent receives $g_{23}$.

Note that this example cannot be modified to act as a counterexample for chores. We have found that in the first part, we cannot sufficiently restrict the possible envy relations. This is due to the fundamental difference in allocating goods and chores: goods cannot be allocated to an agent which is envied, whilst chores cannot be allocated to an agent which envies others.

\section{Omitted Proofs from Section \ref{sec:existence}} \label{app:omitted_existence}
\subsection{Proof of \Cref{thm:twotypes}}
Consider the following greedy algorithm (Algorithm \ref{alg:twotypes}).
    \begin{algorithm}[h!]
    \caption{Returns a TEF1 allocation for goods or chores when there are two types of items}
    \label{alg:twotypes}
    \begin{flushleft} \textbf{Input}: Set of agents $N=\{1,\dots,n\}$, set of items $O=\{o_1,\dots,o_m\}$, and valuation profile $\mathbf{v} = (v_1, \dots, v_n)$ \\
    \textbf{Output}: TEF1 allocation $\mathcal{A}$ of items in $O$ to agents in $N$
    \end{flushleft}
    \begin{algorithmic}[1]
        \STATE Initialize $\alpha \leftarrow 1$,  $\beta \leftarrow n$, and $\mathcal{A}^0 \leftarrow (\varnothing, \dots, \varnothing)$
        \FOR{$t = 1,2,\dots,m$}
        \IF{$\alpha = n+1$}
            \STATE $\alpha \leftarrow 1$
        \ELSIF{$\beta = 0$}
            \STATE $\beta \leftarrow n$
        \ENDIF
        \IF{$o_t \in T_1$}
            \STATE $A_\alpha^t \leftarrow A_\alpha^{t-1} \cup \{o_t\}$, $A_j^t \leftarrow A_j^{t-1}$ for all $j \in N\setminus \{\alpha\}$, and $\alpha \leftarrow \alpha + 1$
        \ELSE
            \STATE  $A_\beta^t \leftarrow A_\beta^{t-1} \cup \{o_t\}$, $A_j^t \leftarrow A_j^{t-1}$ for all $j \in N\setminus \{\beta\}$, and $\beta \leftarrow \beta-1$
        \ENDIF
        \ENDFOR
        \STATE \textbf{return} $\mathcal{A} = (A_1^m, \dots, A_n^m)$
    \end{algorithmic}
\end{algorithm}
    
    The polynomial runtime of the Algorithm \ref{alg:twotypes} is easy to verify: there is only one \textbf{for} loop which runs in $\mathcal{O}(m)$ time, and the other operations within run in $\mathcal{O}(mn)$ time.
    Thus, we focus on proving correctness.

    Intuitively, $\alpha$ and $\beta$ each keep a counter of which agent should be next allocated an item of type $T_1$ and $T_2$, respectively.
    For this reason, for each $r \in \{1,2\}$, we can observe that with respect to only items of type $T_r$, the algorithm allocates these items in a round-robin fashion.
    We can therefore make the following two observations: 
    \begin{enumerate}[(i)]
        \item for any pair of agents $i,j \in N$, if $|A_i^t \cap T_1| > |A_j^t \cap T_1|$, then $i < j$; if $|A_i^t \cap T_2| > |A_j^t \cap T_2|$, then $i > j$; and
        \item for any pair of agents $i,j \in N$, round $t \in [m]$, and $r \in \{1,2\}$, we have that $
        \left| |A_i^t \cap T_r| - |A_j^t \cap T_r| \right| \leq 1$.
    \end{enumerate}
    The first observation follows from fact that the $\alpha$ counter is increasing in agent indices whereas the $\beta$ counter is decreasing in agent indices.
    The second observation follows from the widely-known fact that, with respect to items of a specific type, a round-robin allocation always returns a balanced allocation, i.e., the bundle sizes of any two agents differ by no more than one.
    
    Next, we have that for any two agents $i,j \in N$, round $t \in [m]$, and $r,r' \in \{1,2\}$ where $r \neq r'$, if $|A_i^t \cap T_r| > |A_j^t \cap T_r|$, then $|A_i^t \cap T_{r'}| \leq |A_j^t \cap T_{r'}|$.
    To see this, suppose for a contradiction that there exists agents $i,j \in N$ and round $t \in [m]$ such that both $|A_i^t \cap T_1| > |A_j^t \cap T_1|$ and $|A_i^t \cap T_2| > |A_j^t \cap T_2|$.
    Then, observation (i) will give us $i < j$ and $i > j$ respectively, a contradiction.

    For the case of goods, we have that for any pair of agents $i,j \in N$ and round $t \in [m]$, if $i < j$, then
    \begin{equation} \label{eqn:2types_goods_1}
        v_i(A_i^t \cap T_1) \geq v_i(A_j^t \cap T_1)
    \end{equation}
    because $i$ precedes $j$ in the round-robin allocation order, and by the well-established EF1 property of the round-robin algorithm for goods, there exists a good $g \in A_j^t \cap T_2$ such that 
    \begin{equation} \label{eqn:2types_goods_2}
        v_i(A_i^t \cap T_2) \geq v_i(A_j^t \cap T_2 \setminus \{g\}).
    \end{equation}
    Combining (\ref{eqn:2types_goods_1}) and (\ref{eqn:2types_goods_2}), there exists a good $g \in A_j^t$ such that
    \begin{equation*}
        v_i(A_i^t)  = v_i(A_i^t \cap T_1) + v_i(A_i^t \cap T_2) 
         \geq v_i(A_j^t \cap T_1) + v_i(A_j^t \cap T_2 \setminus \{g\}) 
         = v_i(A_j^t \setminus \{g\}).
    \end{equation*}
    Moreover, if $i > j$, then
    \begin{equation}  \label{eqn:2types_goods_3}
        v_i(A_i^t \cap T_2) \geq v_i(A_i^t \cap T_2),
    \end{equation}
    and there exists a good $g \in A_j^t \cap T_1$ such that
    \begin{equation} \label{eqn:2types_goods_4}
        v_i(A_i^t \cap T_1) \geq v_i(A_j^t \cap T_1 \setminus \{g\}).
    \end{equation}
    Combining (\ref{eqn:2types_goods_3}) and (\ref{eqn:2types_goods_4}), there exists a good $g \in A_j^t$ such that
    \begin{equation*}
        v_i(A_i^t) = v_i(A_i^t \cap T_1) + v_i(A_i^t \cap T_2) 
         \geq v_i(A_j^t \cap T_1 \setminus \{g\}) + v_i(A_j^t \cap T_2) 
         = v_i(A_j^t \setminus \{g\}).
    \end{equation*}
    For the case of chores, we have that for any pair of agents $i,j \in N$ and round $t \in [m]$, if $i > j$, then
    \begin{equation} \label{eqn:2types_chores_1}
        v_i(A_i^t \cap T_1) \geq v_i(A_j^t \cap T_1),
    \end{equation}
    and again by the EF1 property of the round-robin algorithm for chores \citep{AzizCaIg22chores}, there exists a chore $c \in A_i^t \cap T_2$ such that 
    \begin{equation} \label{eqn:2types_chores_2}
        v_i(A_i^t \cap T_2 \setminus \{c\}) \geq v_i(A_j^t \cap T_2).
    \end{equation}
    Combining (\ref{eqn:2types_chores_1}) and (\ref{eqn:2types_chores_2}), there exists a chore $c \in A_i^t$ such that
    \begin{equation*}
        v_i(A^t \setminus \{c\})  = v_i(A_i^t \cap T_1) + v_i(A_i^t \cap T_2 \setminus \{c\}) 
         \geq v_i(A_j^t \cap T_1) + v_i(A_j^t \cap T_2) 
         = v_i(A_j^t).
    \end{equation*}
    Moreover, if $i < j$, then
    \begin{equation}  \label{eqn:2types_chores_3}
        v_i(A_i^t \cap T_2) \geq v_i(A_j^t \cap T_2),
    \end{equation}
    and there exists a chore $c \in A_i^t \cap T_1$ such that
    \begin{equation} \label{eqn:2types_chores_4}
        v_i(A_i^t \cap T_1 \setminus \{c\}) \geq v_i(A_j^t \cap T_1).
    \end{equation}
    Combining (\ref{eqn:2types_chores_3}) and (\ref{eqn:2types_chores_4}), there exists a chore $c \in A_i^t$ such that
    \begin{equation*}
        v_i(A^t \setminus \{c\})  = v_i(A_i^t \cap T_1 \setminus \{c\}) + v_i(A_i^t \cap T_2) 
         \geq v_i(A_j^t \cap T_1) + v_i(A_j^t \cap T_2) 
         = v_i(A_j^t).
    \end{equation*} 
    Thus, our result holds.
    
\subsection{Proof of \Cref{thm:generalizedbinary}}
We first prove the result for goods.
Consider the following greedy algorithm (\Cref{alg:generalized_binary_goods}) which iterates through the rounds and allocates each good to the agent who has the least value for their bundle.

\begin{algorithm}[h!]
    \caption{Returns a TEF1 allocation of goods under generalized binary valuations}
    \label{alg:generalized_binary_goods}
    \begin{flushleft} \textbf{Input}: Set of agents $N=\{1,\dots,n\}$, set of goods $O=\{g_1,\dots,g_m\}$, and valuation profile $\mathbf{v} = (v_1, \dots, v_n)$ \\
    \textbf{Output}: TEF1 allocation of goods $\mathcal{A}$ in $O$ to agents in $N$
    \end{flushleft}
    \begin{algorithmic}[1]
        \STATE Initialize the empty allocation $\mathcal{A}^0$ where $A_i^0 = \varnothing$ for all $i \in N$.
        \FOR{$t = 1,2,\dots,m$}
        \STATE Let $S := \{ i' \in N \ | \ v_{i'}(g_t) > 0 \}$ 
        \IF{$S = \varnothing$}
            \STATE Let $i$ be any agent in $N$
        \ELSE
            \STATE Let $i \in \argmin_{i' \in S} v_{i'}(A^{t-1}_{i'})$, with ties broken arbitrarily
        \ENDIF
        \STATE $A_i^t \leftarrow A_i^{t-1} \cup \{g_t\}$ and $A_j^t \leftarrow A_j^{t-1}$ for all $j \in N\setminus \{i\}$
        \ENDFOR
        \STATE \textbf{return} $\mathcal{A} = (A_1^m, \dots, A_n^m)$
    \end{algorithmic}
\end{algorithm}

We first show that for any $i,j \in N$ and $t \in [m]$, it holds that $v_i(A_i^t) \geq v_j(A_i^t)$.
    Suppose for a contradiction that there exists some $i,j \in N$ and $t \in [m]$ such that $v_i(A_i^t) < v_j(A_i^t)$.
    This means there exists some good $g \in A_i^t$ whereby $v_i(g) = 0$ and $v_j(g) > 0$.
    However, then the algorithm would not have allocated $g$ to $i$, a contradiction.

 Next, we will prove by induction that for every $t \in [m]$, $\mathcal{A}^t$ is TEF1.
 The base case is trivially true: when $t = 1$, if every agent values $g_1$ at $0$, then allocating it to any agent will satisfy TEF1, whereas if some agent values $g_1$, allocating it to any agent will also be TEF1: the envy by any other agent towards this agent will disappear with the removal of $g_1$ from the agent's bundle (every agent's bundle will then be the empty set).

 Then, we prove the inductive step.
 Assume that for some $k\in [m-1]$, $\mathcal{A}^k$ is TEF1.
 We will show that $\mathcal{A}^{k+1}$ is also TEF1. Due to the assumption, it suffices to show that for all $i,j\in N$, there exists a good $g \in A_j^{k+1}$ such that $v_i(A_i^{k+1}) \geq v_i(A_j^{k+1} \setminus \{g\})$. 
 Consider the agent $i \in N$ that is allocated $g_{k+1}$.

We first show agent $i$ must be unenvied before being allocated $g_{k+1}$.
Suppose towards a contradiction this is not the case, i.e., there exists some other agent $j \neq i$ whereby $v_j(A_j^k) < v_j(A_i^k)$.
Together with the fact that $v_j(A_i^k) \leq v_i(A_i^k)$ from the result above, we get that 
\begin{equation*}
    v_j(A_j^k) < v_j(A_i^k) \leq v_i(A_i^k),
\end{equation*}
contradicting the fact that $i$ is an agent with the minimum bundle value and thus chosen by the algorithm to receive $g_{k+1}$.
As such, $i$ must be unenvied before being allocated $g_{k+1}$, i.e., for any other agent $j \in N \setminus \{i\}$, we have that $v_j(A_j^k) \geq v_j(A_i^k)$.

Consequently, we get that
\begin{equation*}
    v_j(A_j^{k+1}) = v_j(A_j^k) \geq v_j(A_i^k) = v_j(A_i^{k+1} \setminus \{g_{k+1}\}).
\end{equation*}
Thus, by induction, the result holds.

Next, we prove the result for chores.
Consider the following greedy algorithm (Algorithm \ref{alg:generalized_binary_chores}) which iterates through the rounds, allocating each chore to an agent with zero value for it if possible, and otherwise, allocates the chore to an agent who does not envy any other agent.

\begin{algorithm}[h!]
    \caption{Returns an TEF1 allocation of chores under generalized binary valuations}
    \label{alg:generalized_binary_chores}
    \begin{flushleft} \textbf{Input}: Set of agents $N=\{1,\dots,n\}$, set of chores $O=\{c_1,\dots,c_m\}$, and valuation profile $\mathbf{v} = (v_1, \dots, v_n)$ \\
    \textbf{Output}: TEF1 allocation of chores $\mathcal{A}$ in $O$ to agents in $N$ 
    \end{flushleft}
    \begin{algorithmic}[1]
        \STATE Initialize the empty allocation $\mathcal{A}^0$ where $A_i^0 = \varnothing$ for all $i \in N$.
        \FOR{$t = 1,2,\dots,m$}
        \IF{there exists an agent $i \in N$ such that $v_i(c_t) = 0$}
            \STATE Let $i \in \{i' \in N \ | \ v_{i'}(c_t) = 0\}$ 
        \ELSE
            \STATE Let $i \in \argmax_{i' \in N}  v_{i'}(A^{t-1}_{i'})$\\
        \ENDIF
        \STATE $A_i^t \leftarrow A_i^{t-1} \cup \{c_t\}$ and $A_j^t \leftarrow A_j^{t-1}$ for all $j \in N\setminus \{i\}$
        \ENDFOR
        \STATE \textbf{return} $\mathcal{A} = (A_1^m, \dots, A_n^m)$
    \end{algorithmic}
\end{algorithm}

We first show that for any $i,j \in N$ and $t \in [m]$, it holds that 
\begin{equation} \label{eqn:genbin_contradiction}
    v_i(A_i^t) \geq v_j(A_i^t).
\end{equation} 
Suppose for a contradiction that there exists some $i,j\in N$ and $t\in [m]$ such that $v_i(A_i^t) < v_j(A_i^t)$. 
This means there exists some chore $c\in A_i^t$ whereby $v_i(c) < 0$ and $v_j(c) = 0$.
However, then the algorithm would not have allocated $c$ to $i$, a contradiction. 

Next, we will prove by induction that for every $t \in [m]$, $\mathcal{A}^t$ is TEF1.
The base case is trivially true: when $t=1$, if there exists an agent that values $c_1$ at $0$, then allocating it to any such agent will satisfy TEF1, whereas if all agents values $c_1$ negatively, allocating it to any agent will also be TEF1: the envy by this agent towards any other agent will disappear with the removal of $c_1$ from the former agent's bundle (every agent's bundle will then be the empty set).

Then, we prove the inductive step. 
Assume that for some $k\in [m-1]$, $\mathcal{A}^k$ is TEF1. 
We will show that $\mathcal{A}^{k+1}$ is also TEF1, i.e., for all $i,j\in N$, there exists a chore $c \in A_i^{k+1}$ such that $v_i(A_i^{k+1}\setminus \{c\}) \geq v_i(A_j^{k+1})$.

Suppose agent $i$ is allocated the chore $c_{k+1}$. If $v_i(c_{k+1})=0$, then each agents' valuation for every other agent's bundle (including his own) remains the same, and thus $\mathcal{A}^{k+1}$ remains TEF1. 
If $v_i(c_{k+1})<0$, then we know that $v_j(c_{k+1})<0$ for all $j\in N$. 
We then proceed to show that agent $i$ must not envy any other agent before being allocated $c_{k+1}$. Suppose for contradiction this is not the case, i.e., that there exists some other agent $j\neq i$ whereby $v_i(A_i^k) < v_i(A_j^k)$. 
Since $c_{k+1}$ is allocated to the agent with the highest bundle, we have that $v_i(A^k_i) \geq v_j(A^k_j)$, and therefore
\begin{equation*}
    v_i(A_j^k)>v_i(A_i^k)\geq v_j(A^k_j).
\end{equation*}
However, this contradicts (\ref{eqn:genbin_contradiction}).

Since agent $i$ does not envy another agent before being allocated $c_{k+1}$, we get that for any $j\neq i$,
\begin{equation*}
    v_i(A_i^{k+1}\setminus \{c_{k+1}\}) = v_i(A_i^k) \geq v_i(A_j^k)=v_i(A_j^{k+1}) \quad \text{and} \quad v_j(A_j^{k+1})=v_j(A_j^k).
\end{equation*}
Thus, by induction, we get that $\mathcal{A}^{t+1}$ is TEF1.

\subsection{Proof of \Cref{thm:singlepeaked_goods}}
Consider the following greedy algorithm (Algorithm \ref{alg:singlepeaked_items}).
Note that the same algorithm works for both settings for goods when valuations are single-peaked, and for chores when valuations are single-dipped. 

    \begin{algorithm}[H]
    \caption{Returns a TEF1 allocation for goods when valuations are single-peaked and chores when valuations are single-dipped}
    \label{alg:singlepeaked_items}
    \begin{flushleft} \textbf{Input}: Set of agents $N=\{1,\dots,n\}$, set of items $O=\{o_1,\dots,o_m\}$, and valuation profile $\mathbf{v} = (v_1, \dots, v_n)$ \\
    \textbf{Output}: TEF1 allocation $\mathcal{A}$ of items in $O$ to agents in $N$
    \end{flushleft}
    \begin{algorithmic}[1]
        \STATE Initialize $\mathcal{A}^0 \leftarrow (\varnothing, \dots, \varnothing)$
        \FOR{$t = 1,2,\dots,m$}
        \STATE Let $i := \argmin_{i \in N} |A_i^{t-1}|$, with ties broken lexicographically
        \STATE $A^t_i \leftarrow A^{t-1}_i \cup \{g_t\}$ and $A^t_i \leftarrow A^{t-1}_i$
        \ENDFOR
        \STATE \textbf{return} $\mathcal{A} = (A_1^m, \dots, A_n^m)$
    \end{algorithmic}
\end{algorithm}

    The polynomial runtime of the Algorithm \ref{alg:singlepeaked_items} is easy to verify: there is only one \textbf{for} loop which runs in $\mathcal{O}(m)$ time, and the other operations within run in $\mathcal{O}(n)$ time.
    Thus, we focus on proving correctness.

We first prove the case for goods, when valuations are single-peaked.

    For each $i \in [m]$, let $g_i =  o_i$, and thus $O = \{g_1,\dots,g_m\}$.
    We can assume that $m = \alpha n$ for some $\alpha \in \mathbb{Z}_{>0}$; otherwise we can simply add dummy goods to $O$ until that condition is fulfilled.
    Then, Algorithm~\ref{alg:singlepeaked_items} will return $\mathcal{A}$, where for each $i \in N$, $A_i = \{g_i, g_{i+n}, \dots, g_{i+(\alpha-1)n}$\}.

    For each $i \in N$ and $j \in [\alpha]$, let 
    \begin{itemize}
        \item $T_j := \{g_{(j-1)n + 1}, g_{(j-1)n + 2}, \dots, g_{jn}\}$,
        \item $g'_{i,j} \in A_i \cap T_j$ be the unique good in $T_j$ that was allocated to agent~$i$
        \item $g^* := \argmax_{g \in O} v_i(g)$ (with ties broken arbitrarily), and $g^* \in T_{i^*}$ for some $i^* \in [\alpha]$.
    \end{itemize}

    Then, we will show that for all $r \in [\alpha]$, $v_i(A_i^r) \geq v_i(A_j^r \setminus \{g\})$ for some $g \in A_j^r$.
    We split our analysis into two cases.
    \begin{description}
        \item[Case 1: $i < j$.]
        If $r < i^*$, then since agent $i$'s valuation for each subsequent good up to round $T_r$ is non-decreasing, we have that for all $k\in \{2,\dots,r\}$, 
        \begin{equation*}
            v_i(g'_{i,k}) \geq v_i(g'_{j,k-1}).
        \end{equation*}
        Consequently, we get that 
        \begin{equation*}
            v_i(A_i^r) \geq \sum_{k=2}^r v_i(g'_{i,k}) \geq \sum_{k=2}^r v_i(g'_{j,k-1})= v_i(A_j^r \setminus \{g'_{j,r}\}). 
        \end{equation*}
        If $r \geq i^*$, then we split our analysis into two further cases.

        \begin{description}
            \item[Case 1(a): $g'_{i,i^*}$ appears before $g^*$.] 
            Then, for all $k \in \{2,\dots, i^*\}$, 
            \begin{equation*}
                v_i(g'_{i,k}) \geq v_i(g'_{j,k-1}),
            \end{equation*}
            and for all $k \in \{i^*+1,\dots,r\}$,
            \begin{equation*}
                v_i(g'_{i,k}) \geq v_i(g'_{j,k}).
            \end{equation*}
            Consequently, we get that
            \begin{equation*}
                v_i(A_i^r) \geq \sum_{k=2}^{i^*} v_i(g'_{i,k}) + \sum_{k=i^*+1}^r v_i(g'_{i,k}) \geq \sum_{k=2}^{i^*} v_i(g'_{j,k-1}) + \sum_{k=i^*+1}^r v_i(g'_{j,k})  = v_i(A_j^r \setminus \{g'_{j,i^*}\}).
            \end{equation*}

        \item[Case 1(b): $g'_{i,i^*}$ appears after (or is) $g^*$.]
        Then, for all $k \in \{2,\dots,i^*-1\}$, 
        \begin{equation*}
            v_i(g'_{i,k}) \geq v_i(g'_{j,k-1}),
        \end{equation*}
        and for all $k \in \{i^*,\dots ,r\}$,
        \begin{equation*}
            v_i(g'_{i,k}) \geq v_i(g'_{j,k}).
        \end{equation*}
        Consequently, we get that
        \begin{align*}
                v_i(A_i^r) & \geq \sum_{k=2}^{i^*-1} v_i(g'_{i,k}) + \sum_{k=i^*}^r v_i(g'_{i,k})\geq \sum_{k=2}^{i^*-1} v_i(g'_{j,k-1}) + \sum_{k=i^*}^r v_i(g'_{j,k}) = v_i(A_j^r \setminus \{g'_{j,i^*-1}\}).
            \end{align*}
        \end{description}
        \item[Case 2: $i > j$.]
        If $r \leq i^*$, then since agent $i$'s valuation for each subsequent good up to round $T_r$ is nondecreasing, we have that for all $k \in [r]$,
        \begin{equation} \label{eqn:single_peaked_case2_1}
            v_i(g'_{i,k}) \geq v_i(g'_{j,k}).
        \end{equation}
        Consequently, we get that
        \begin{equation*}
            v_i(A_i^r) \geq \sum_{k \in [r-1]} v_i(g'_{i,k}) \geq \sum_{k \in [r-1]} v_i(g'_{j,k}) \quad \text{(by (\ref{eqn:single_peaked_case2_1}))} = v_i(A_j^r \setminus \{g'_{j,r}\}).
        \end{equation*}
        If $r > i^*$, then we split our analysis into two further cases.
        \begin{description}
            \item[Case 2(a): $g'_{i,i^*}$ appears before (or is) $g^*$.]
            Then for all $k \in [i^*]$,
            \begin{equation*}
                v_i(g'_{i,k}) \geq v_i(g'_{j,k}),
            \end{equation*}
            and for all $k \in \{i^*+1,\dots, r-1\}$,
            \begin{equation*}
                v_i(g'_{i,k}) \geq v_i(g'_{j,k+1}).
            \end{equation*}
            Consequently, we get that
            \begin{equation*}
                v_i(A_i^r) \geq \sum_{k \in [i^*]} v_i(g'_{i,k}) + \sum_{k= i^*+1}^{r-1} v_i(g'_{i,k}) \geq \sum_{k \in [i^*]} v_i(g'_{j,k}) + \sum_{k = i^*+1}^{r-1} v_i(g'_{j,k+1})  = v_i(A_j^r \setminus \{g'_{j,i^*+1}\}).
            \end{equation*}
            \item[Case 2(b): $g'_{i,i^*}$ appears after $g^*$.]
            Then, for all $k \in [i^*-1]$,
            \begin{equation*}
                v_i(g'_{i,k}) \geq v_i(g'_{j,k}),
            \end{equation*}
            and for all $k \in \{i^*,\dots, r-1\}$,
            \begin{equation*}
                v_i(g'_{i,k}) \geq v_i(g'_{j,k+1}).
            \end{equation*}
            Consequently, we get that
            \begin{equation*}
                v_i(A_i^r) \geq \sum_{k \in [i^*-1]} v_i(g'_{i,k}) + \sum_{k =i^*}^{r-1} v_i(g'_{i,k}) \geq \sum_{k \in [i^*-1]} v_i(g'_{j,k}) + \sum_{k =i^*}^{r-1} v_i(g'_{j,k+1}) = v_i(A_j^r \setminus \{g'_{j,i^*}\}).
            \end{equation*}
        \end{description}
    \end{description}
    Thus, our result follows.

    Next, we prove the case for chores, when valuations are single-dipped.
    
    For each $j \in [m]$, let $o_i = c_i$, and thus $O = \{c_1,\dots,c_m\}$.
    We can assume that $m = \alpha n$ for some $\alpha \in \mathbb{Z}_{>0}$; otherwise we can simply add dummy chores to $O$ until that condition is fulfilled.
    Then, Algorithm~\ref{alg:singlepeaked_items} will return $\mathcal{A}$, where for each $i \in N$, $A_i = \{c_i, c_{i+n}, \dots, c_{i+(\alpha-1)n}$\}.

    For each $i \in N$ and $j \in [\alpha]$, let 
    \begin{itemize}
        \item $T_j := \{c_{(j-1)n + 1}, c_{(j-1)n + 2}, \dots, c_{jn}\}$,
        \item $c'_{i,j} \in A_i \cap T_j$ be the unique chore in $T_j$ that was allocated to agent~$i$
        \item $c^* := \argmin_{c \in O} v_i(c)$ (with ties broken arbitrarily), and $c^* \in T_{i^*}$ for some $i^* \in [\alpha]$.
    \end{itemize}

        \begin{description}
        \item[Case 1: $i < j$.]
        If $r \leq j^*$, then since agent $i$'s valuation for each subsequent chore up to round $T_{r-1}$ is nonincreasing, we have that for all $k \in [r-1]$,
        \begin{equation*}
            v_i(c'_{i,k}) \geq v_i(c'_{j,k}).
        \end{equation*}
        Consequently, we get that 
        \begin{equation*}
            v_i(A_i^r \setminus \{c'_{i,r}\}) = \sum_{k \in [r-1]} v_i(c'_{i,k}) \geq \sum_{k \in [r-1]} v_i(c'_{j,k}) \geq v_i(A_j^r).
        \end{equation*}

        If $r > j^*$, then we split our analysis into two further cases.
        \begin{description}
            \item[Case 1(a): $c'_{j,j^*}$ appears before (or is) $c^*$.]
            Then for all $k \in [j^*]$,
            \begin{equation*}
                v_i(c'_{i,k}) \geq v_i(c'_{j,k})
            \end{equation*}
            and for all $k \in \{j^*+2,\dots, r\}$,
            \begin{equation*}
                v_i(c'_{i,k}) \geq v_i(c'_{j,k-1}).
            \end{equation*}
            Consequently, we get that 
            \begin{equation*}
                v_i(A_i^r \setminus \{c'_{i,j^*+1}\}) = \sum_{k \in [j^*]} v_i(c'_{i,k}) + \sum_{k =j^*+2}^r v_i(c'_{i,k}) \geq \sum_{k \in [j^*]} v_i(c'_{j,k}) + \sum_{k =j^*+2}^r v_i(c'_{j,k-1}) \geq v_i(A_j^r).
            \end{equation*}
            \item[Case 1(b): $c'_{j,j^*}$ appears after $c^*$.]
            Then for all $k \in [j^*-1]$,
            \begin{equation*}
                v_i(c'_{i,k}) \geq v_i(c'_{j,k})
            \end{equation*}
            and for all $k \in \{j^*+1,\dots, r\}$,
            \begin{equation*}
                v_i(c'_{i,k}) \geq v_i(c'_{j,k-1}).
            \end{equation*}
            Consequently, we get that
            \begin{equation*}
                v_i(A_i^r \setminus \{c'_{i,j^*}) = \sum_{k \in [1,j^*-1]} v_i(c'_{i,k}) + \sum_{k = j^*+1}^r v_i(c'_{i,k}) \geq \sum_{k \in [j^*-1]} v_i(c'_{j,k}) + \sum_{k = j^*+1}^r v_i(c'_{j,k-1}) \geq v_i(A_j^r).
            \end{equation*}
        \end{description}
        \item[Case 2: $j < i$.]
        If $r < j^*$, then since agent $i$'s valuation for each subsequent chore up to round $T_r$ is nondecreasing, we have that for all $k \in [r-1]$,
        \begin{equation*}
            v_i(c'_{i,k}) \geq v_i(c'_{j,k+1}).
        \end{equation*}
        Consequently, we get that 
        \begin{equation*}
            v_i(A_i^r \setminus \{c'_{i,r}\})= \sum_{k \in [r-1]} v_i(c'_{i,k}) \geq \sum_{k \in [r-1]} v_i(c'_{j,k+1}) \geq v_i(A_j^r). 
        \end{equation*}
        If $r \geq j^*$, then we split our analysis into two further cases.
        \begin{description}
            \item[Case 2(a): $c_{j,j^*}$ appears before (or is) $c^*$.]
            Then for all $k \in [j^*-1]$, 
            \begin{equation*}
                v_i(c'_{i,k}) \geq v_i(c'_{j,k+1})
            \end{equation*}
            and for all $k \in \{j^*+1,\dots, r\}$, 
            \begin{equation*}
                v_i(c'_{i,k}) \geq v_i(c'_{j,k}).
            \end{equation*}
            Consequently, we get
            \begin{equation*}
                v_i(A_i^r \setminus \{c'_{i,j^*}\}) = \sum_{k \in [j^*-1]} v_i(c'_{i,k}) + \sum_{k = j^*+1}^r v_i(c'_{i,k}) \geq \sum_{k \in [j^*-1]} v_i(c'_{j,k+1}) + \sum_{k = j^*+1}^r v_i(c'_{k,j}) \geq v_i(A_j^r).
            \end{equation*}
        \end{description}
        \item[Case 2(b): $c_{j,j^*}$ appears after $c^*$.]
        Then for all $k \in [j^*-2]$,
        \begin{equation*}
            v_i(c'_{i,k}) \geq v_i(c'_{j,k+1}
        \end{equation*}
        and for all $k \in \{j^*,\dots,r\}$,
        \begin{equation*}
            v_i(c'_{i,k}) \geq v_i(c'_{j,k}).
        \end{equation*}
        Consequently, we get that
        \begin{equation*}
            v_i(A_i^r \setminus \{c'_{i,j^*-1}\}) = \sum_{k \in [j^*-2]} v_i(c'_{i,k}) + \sum_{k =j^*}^r v_i(c'_{i,k}) \geq \sum_{k \in [j^*-2]} v_i(c'_{j,k+1}) + \sum_{k =j^*}^r v_i(c'_{j,k}) \geq v_i(A_j^r).
        \end{equation*}
    \end{description}
    Thus, our result follows.

\subsection{Proof of Lemma \ref{lemma-code}} \label{app:code}
\begin{lstlisting}
from itertools import combinations
from copy import deepcopy

# If there exists a partial allocation for first 2n+2 rounds such that bundle valuations are equal
if_some_envy_exists = False

def is_ef1(allocation, agents, valuations, if_some_envy_exists, partial_alloc_envy_from, partial_alloc_envy_to):
    """
    Check if the current allocation is EF1.

    Parameters:
    - allocation: List of lists, where allocation[i] is the list of goods allocated to agent i.
    - agents: List of agent identifiers.
    - valuations: Dictionary where valuations[agent][good] gives the value of a good for an agent.
    - if_some_envy_exists: If the partial allocation for the first 2n+2 rounds is EF (i.e., equal bundle values)
    - partial_alloc_envy_from: If if_some_envy_exists is True, then which agent envies
    - partial_alloc_envy_to: If if_some_envy_exists is True, then which agent is being envied

    Returns:
    - True if allocation is EF1, False otherwise.
    """
    num_agents = len(agents)

    # Compute the value each agent has for their own bundle
    agent_own_values = []
    for agent_idx in range(num_agents):
        agent = agents[agent_idx]
        total = sum(valuations[agent][good] for good in allocation[agent_idx])
        agent_own_values.append(total)

    # Check EF1 condition for every pair of agents (i, j)
    for i in range(num_agents):
        for j in range(num_agents):
            if i == j:
                continue
            agent_i = agents[i]
            agent_j_bundle = allocation[j]
            # Agent i's value for agent j's bundle
            lst = [valuations[agent_i][good] for good in agent_j_bundle]
            if lst:
                max_value = max(lst)
            else:
                max_value = 0
            value_i_for_j_less_one = sum(lst) - max_value
            # Agent i's own value
            value_i_own = agent_own_values[i]

            if if_some_envy_exists:
              if i == partial_alloc_envy_from:
                if j == partial_alloc_envy_to:
                  value_i_for_j_less_one += 1

            if value_i_own < value_i_for_j_less_one:
                return False
    return True

def find_ef1_allocations(agents, goods, valuations, if_some_envy_exists, partial_alloc_envy_from=0, partial_alloc_envy_to=0):
    """
    Find all allocations that are EF1 at each step of allocating goods one by one.

    Parameters:
    - agents: List of agent identifiers.
    - goods: List of goods to be allocated.
    - valuations: Dictionary where valuations[agent][good] gives the value of a good for an agent.
    - if_some_envy_exists: If the partial allocation for the first 2n+2 rounds is EF (i.e., equal bundle values)
    - partial_alloc_envy_from: If if_some_envy_exists is True, then which agent envies
    - partial_alloc_envy_to: If if_some_envy_exists is True, then which agent is being envied
    
    Returns:
    - List of allocations. Each allocation is a list of lists, where allocation[i] is the list of goods for agent i.
    """
    num_agents = len(agents)
    all_allocations = []

    def backtrack(current_allocation, index):
        """
        Recursive helper function to perform backtracking.

        Parameters:
        - current_allocation: Current allocation state.
        - index: Index of the next good to allocate.
        """
        if index == len(goods):
            # All goods allocated, add to results
            all_allocations.append(deepcopy(current_allocation))
            return

        current_good = goods[index]

        for agent_idx in range(num_agents):
            # Assign current_good to agent_idx
            current_allocation[agent_idx].append(current_good)

            # Check EF1 condition at this step
            if is_ef1(current_allocation, agents, valuations, if_some_envy_exists, partial_alloc_envy_from, partial_alloc_envy_to):
                # Continue to allocate the next good
                backtrack(current_allocation, index + 1)

            # Backtrack: remove the good from the agent's allocation
            current_allocation[agent_idx].pop()

    # Initialize allocation: list of empty lists for each agent
    initial_allocation = [[] for _ in agents]
    backtrack(initial_allocation, 0)

    return all_allocations

# Example Usage
if __name__ == "__main__":
    # Define agents and goods
    agents = ['A', 'B','C']
    goods = ['g1', 'g2', 'g3','g4', 'g5','g6', 'g7','g8', 'g9','g10', 'g11','g12', 'g13','g14', 'g15','g16', 'g17','g18', 'g19','g20','g21']

    # Define valuations for each agent
    valuations = {
        'A': {'g1': 90, 'g2': 80, 'g3': 70, 'g4' : 100, 'g5' : 100,'g6': 100, 'g7':15,'g8':10000, 'g9':11000,'g10':12000, 'g11':20000,'g12':20000, 'g13':20000,'g14':20000, 'g15':20000,'g16':20000, 'g17':20000,'g18':20000, 'g19':20000,'g20':19010, 'g21' :18005},
      'B': {'g1': 90, 'g2': 70, 'g3': 80, 'g4' : 100, 'g5' : 100,'g6': 100, 'g7':95,'g8':10000, 'g9':11000,'g10':12000, 'g11':20000,'g12':20000, 'g13':20000,'g14':20000, 'g15':20000,'g16':20000, 'g17':20000,'g18':12000, 'g19':12000,'g20':19085, 'g21' :14106},
      'C': {'g1': 80, 'g2': 90, 'g3': 70, 'g4' : 100, 'g5' : 100,'g6': 100, 'g7':25,'g8':10000, 'g9':11000,'g10':12000, 'g11':20000,'g12':20000, 'g13':18500,'g14':20000, 'g15':20000,'g16':20000, 'g17':20000,'g18':20000, 'g19':20000,'g20':19010, 'g21' :19496}
    }

    # Find all EF1 allocations
    if if_some_envy_exists:
        for partial_alloc_envy_from in range(3):
            for partial_alloc_envy_to in range(3):
                if partial_alloc_envy_from != partial_alloc_envy_to:
                    ef1_allocations = find_ef1_allocations(agents, goods, valuations,True, partial_alloc_envy_from, partial_alloc_envy_to)
                    # Each iteration considers different combinations of envy that exists in the partial allocation for the first 2n+2 rounds
                    
                    # Print the allocations
                    print(f"Total EF1 allocations: {len(ef1_allocations)}\n")
    else:
        ef1_allocations = find_ef1_allocations(agents, goods, valuations, False)
        # Print the allocations
        print(f"Total EF1 allocations: {len(ef1_allocations)}\n")
    for idx, alloc in enumerate(ef1_allocations, 1):
        print(f"Allocation {idx}:")
        for agent_idx, agent in enumerate(agents):
            print(f"  {agent}: {alloc[agent_idx]}")
        print(f" {sum(valuations['A'][good] for good in alloc[0]) - sum(valuations['A'][good] for good in alloc[1])},{sum(valuations['A'][good] for good in alloc[0]) - sum(valuations['A'][good] for good in alloc[2])}")
        print(f" {sum(valuations['B'][good] for good in alloc[1]) - sum(valuations['B'][good] for good in alloc[0])},{sum(valuations['B'][good] for good in alloc[1]) - sum(valuations['B'][good] for good in alloc[2])}")
        print(f" {sum(valuations['C'][good] for good in alloc[2]) - sum(valuations['C'][good] for good in alloc[0])},{sum(valuations['C'][good] for good in alloc[2]) - sum(valuations['C'][good] for good in alloc[1])}")
        print()

\end{lstlisting}
        
\section{Omitted Proofs from Section \ref{sec:efficiency}} \label{app:omitted_efficiency}
\subsection{Proof of \Cref{thm:tef1_po_nphard_goods}}
We reduce from the NP-hard problem \textsc{1-in-3-SAT}.
    An instance of this problem consists of conjunctive normal form $F$ with three literals per clause; it is a yes-instance if there exists a truth assignment to the variables such that each clause has exactly one \texttt{True} literal, and a no-instance otherwise.
    
    Consider an instance of \textsc{1-in-3-SAT} given by the CNF $F$ which contains $n$ variables $\{x_1, \dots, x_n\}$ and $m$ clauses $\{C_1, \dots, C_m\}$.

    We construct an instance $\mathcal{I}$ with two agents and $2n+1$ goods.
    For each $i \in [n]$, we introduce two goods $t_i$, $f_i$. 
    We also introduce an additional good $r$.
    Let agents' (identical) valuations be defined as follows:
    \begin{equation*}
        v(g) = 
        \begin{cases} 
            5^{m+n-i} + \sum_{j \, : \, x_i \in C_j} 5^{m-j}, & \text{if } g = t_i, \\ 
            5^{m+n-i} + \sum_{j \, : \, \neg x_i \in C_j}  5^{m-j}, & \text{if } g = f_i, \\
            \sum_{j \in[m]} 5^{j-1}, & \text{if } g = r.
        \end{cases}
    \end{equation*}
    Intuitively, for each variable index $i \in [n]$, we associate with it a unique value $5^{m+n-i}$.
    For each clause index $j \in [m]$, we also associate with it a unique value $5^{m-j}$.
    Note that no two indices (regardless of whether its a variable or clause index) share the same value, hence the uniqueness of the values.
    Then, the value for each good $t_i$ comprises of the unique value associated with $i$, and the sum over all unique values of clauses $C_j$ which $x_i$ appears as a \emph{positive literal} in; whereas the value for each good $f_i$ comprises of the unique value associated with $i$, and the sum over all unique values of clauses $C_j$ which $x_i$ appears as a \emph{negative literal} in.
    We will utilize this in our analysis later.

    Then, we have the set of goods $O = \{t_1,f_1,t_2,f_2,\dots,t_n,f_n,r\}$.
    Note that
    \begin{equation*}
        v(O) = v(r) + \sum_{i \in [n]} v(t_i) + \sum_{i \in [n]} v(f_i).
    \end{equation*}
    Also observe that
    \begin{equation*}
        \sum_{i \in [n]} 5^{m+n-i} = \sum_{i \in [n]} 5^{m+i-1}.
    \end{equation*}
    Now, as each clause contains exactly three literals,
    \begin{equation*}
        \sum_{i \in [n]}\sum_{j : x_i \in C_j} 5^{m-j} + \sum_{i \in [n]}\sum_{j : \neg x_i \in C_j} 5^{m-j} = 3 \times \sum_{j \in [m]} 5^{j-1}.
    \end{equation*}
    Then, combining the equations above, we get that
    \begin{equation} \label{eqn:lem_vO_goods}
        v(O) = 2 \times \sum_{i \in [n]} 5^{m+i-1} + 4 \times \sum_{j \in [m]} 5^{j-1}.
    \end{equation}
        
    Let the goods appear in the following order: 
    \begin{equation*}
        t_1,f_1,t_2,f_2,\dots,t_n,f_n,r.
    \end{equation*}
    We first prove the following result.
    \begin{lemma}\label{lem:potef1goodshard}
        There exists a truth assignment $\alpha$ such that each clause in $F$ has exactly one \texttt{True} literal 
        if and only if 
        there exists an 
        allocation $\mathcal{A} = (A_1$, $A_2)$ such that $v(A_1) = v(A_2)$ for instance $\mathcal{I}$.
    \end{lemma}
    \begin{proof}
        For the `if' direction, consider an allocation $\mathcal{A}$ such that $v(A_1) = v(A_2)$. 
        Since agents have identical valuations, without loss of generality, let $r \in A_1$. 
        Since $O = A_1 \cup A_2$ and $v(A_1) = v(A_2) = \frac{1}{2}v(O)$, we have that
        \begin{equation*}
            v(A_1 \setminus \{r\}) = \left( \sum_{i \in [n]} 5^{m + i -1} + 2 \times \sum_{j \in [m]} 5^{j - 1} \right) - \sum_{j \in [m]} 5^{j - 1}
            = \sum_{i \in [n]} 5^{m + i -1} + \sum_{j \in [m]} 5^{j - 1}.
        \end{equation*} 
        Note that this is only possible if for each $i \in [m]$, $t_i$ and $f_i$ are allocated to different agents.
        The reason is because the only way agent~$1$ can obtain the $\sum_{i \in [n]} 5^{m + i -1}$ term of the above bundle value is if he is allocated exactly one good from each of $\{t_i,f_i\}$ for all $i \in [n]$.

        Then, from the goods that exist in bundle $A_1$, we can construct an assignment $\alpha$:
        for each $i \in [n]$, let $x_i= \texttt{True}$ if $t_i \in A_1$ and $x_i = \texttt{False}$ if $f_i \in A_1$.
        Then, from the second term in the expression of $v(A_1 \setminus\{r\})$ above, we can observe that each clause must have exactly one \texttt{True} literal.

        For the `only if' direction,  consider a truth assignment $\alpha$ such that each clause in $F$ has exactly one \texttt{True} literal. 
        Then, for each $i \in [n]$, let 
        \begin{equation*}
            \ell_i = 
            \begin{cases} 
                t_i & \text{if } x_i = \texttt{True} \text{ under } \alpha, \\
                f_i  & \text{if } x_i = \texttt{False} \text{ under } \alpha. 
            \end{cases}
        \end{equation*}
        We construct the allocation $\mathcal{A} = (A_1,A_2)$ where 
        \begin{equation*}
            A_1 = \{\ell_1, \dots, \ell_n, r\} \quad \text{and} \quad A_2 = O \setminus A_1.
        \end{equation*}
        Again, observe that
        \begin{equation*}
            \sum_{i \in [n]} 5^{m+n-i} = \sum_{i \in [n]} 5^{m+i-1}.
        \end{equation*}
        Then, as each clause has exactly one \texttt{True} literal, 
        \begin{equation*}
            v(A_1) = \sum_{i \in [n]}  5^{m + i -1} + 2 \times \sum_{j \in [m]} 5^{j - 1},
        \end{equation*}
        and together with (\ref{eqn:lem_vO_goods}), we get that
        \begin{equation*}
            v(A_2) = v(O) - v(A_1) = v(A_1),
        \end{equation*}
        as desired.
    \end{proof}
    Note that for all values of $m,n \geq 1$, and some $\varepsilon < \frac{1}{3}$,
    \begin{equation} \label{eqn:tef1po_kcondition_goods_2}
        5^{m+n} - 2\varepsilon > 5^{m+n-1} + \frac{5^m-1}{4} = 5^{m+n-1} + \sum_{j \in [m]} 5^{j-1} \geq \max_{g \in O} v(g).
    \end{equation}
    
    Now consider another instance $\mathcal{I}'$ that is similar to $\mathcal{I}$, but with an additional four goods $o_1,o_2,o_3,o_4$.
    Let the agents' valuations over these four new goods be defined as follows, for some $\varepsilon < \frac{1}{3}$:
    \begin{center}
        \begin{tabular}{l||cccc}
            $\mathbf{v}$ &$o_1$ & $o_2$ & $o_3$ & $o_4$\\ \hline\hline
            1 & $5^{m+n}$ & $5^{m+n} - \varepsilon$ & $5^{m+n} - \varepsilon$ & $5^{m+n}$\\
            2 & $5^{m+n}-\varepsilon$ & $5^{m+n}$ & $5^{m+n}$ & $5^{m+n}-\varepsilon$\\
        \end{tabular}
    \end{center}
    Then, we have the set of goods $O' = O \cup\{o_1,o_2,o_3,o_4\}$.
    
    Let the goods be in the following order:
    \begin{equation*}
        t_1,f_1,t_2,f_2,\dots,t_n,f_n,r, o_1,o_2,o_3,o_4.
    \end{equation*}
    If there is a partial allocation $\mathcal{A}^{2n+1}$ over the first $2n+1$ goods
    such that $v(A^{2n+1}_1) = v(A^{2n+1}_2)$, then by giving $o_1,o_4$ to agent 1 and $o_2,o_3$ to agent 2, we obtain an allocation that is TEF1 and PO (note that any allocation for the first $2n+1$ goods will be PO, since agents have identical valuations over them). 
    
    However, if there does not exist a partial allocation $\mathcal{A}^{2n+1}$ over the first $2n+1$ goods
    such that $v(A^{2n+1}_1) = v(A^{2n+1}_2)$, then let $\mathcal{A}^{2n+1}$ be any partial allocation of the first $2n+1$ goods that is TEF1 but $v(A^{2n+1}_1) \neq v(A^{2n+1}_2)$. We will show that if $v(A^{2n+1}_1) \neq v(A^{2n+1}_2)$, any TEF1 allocation of $O'$ cannot be PO.

    Note that in order for $\mathcal{A}^{2n+1}$ to be TEF1, we must have that for any agent $i \in \{1,2\}$,
    \begin{equation} \label{eqn:tef1po_kcondition_goods_3}
        v(A^{2n+1}_i) \geq \frac{v(O) - \max_{g \in O} v(g)}{2}.
    \end{equation}
    This also means that for any agent $i \in \{1,2\}$,
    \begin{equation}\label{eqn:tef1po_kcondition_goods_4}
        v(A^{2n+1}_i) \leq v(O) - \frac{v(O) - \max_{g \in O} v(g)}{2} = \frac{v(O) + \max_{g \in O} v(g)}{2}.
    \end{equation}
    Also observe that since $\min_{g \in O} v(g) > \varepsilon$ and $v(A^{2n+1}_1) \neq v(A^{2n+1}_2)$,  
    \begin{equation}\label{eqn:tef1po_kcondition_goods_5}
        \left| v(A^{2n+1}_1) - v(A^{2n+1}_2) \right| > \varepsilon.
    \end{equation}

    We split our analysis into two cases.

    \begin{description}
        \item[Case 1: $v(A^{2n+1}_1) > v(A^{2n+1}_2)$.] 
    If we give $o_1$ to agent~$1$, since by (\ref{eqn:tef1po_kcondition_goods_2}), $v_2(o_1) > \max_{g \in A^{2n+1}_1} v(g)$, we get that
    \begin{equation*}
        v_2(A^{2n+2}_2) = v(A^{2n+1}_2) <  v(A^{2n+1}_1) = v_2(A^{2n+2}_1 \setminus \{o_1\}),
    \end{equation*}
    and agent~$2$ will still envy agent~$1$ after dropping $o_1$ from agent~$1$'s bundle.
    Thus, we must give $o_1$ to agent~$2$.

    Next, if we give $o_2$ to agent~$2$, then since $v_1(o_1) > \max_{g \in O} v(g)$ and $v_1(o_1) > v_1(o_2)$, we have that
    \begin{align*}
        v_1(A_1^{2n+3}) & = v(A_1^{2n+1}) \\
        & \leq \frac{v(O) + \max_{g \in O} v(g)}{2} \quad (\text{by } (\ref{eqn:tef1po_kcondition_goods_4}))\\
        & < \frac{v(O) - \max_{g \in O} v(g)}{2} + v_1(o_2) \quad (\text{by } (\ref{eqn:tef1po_kcondition_goods_2})) \\
        & \leq v(A_2^{2n+1}) + v_1(o_2) \quad (\text{by } (\ref{eqn:tef1po_kcondition_goods_3}))\\
        & = v_1(A_2^{2n+3} \setminus \{o_1\}),
    \end{align*}
    and agent~$1$ will still envy agent~$2$ after dropping $o_1$ from agent~$2$'s bundle.
    Thus, we must give $o_2$ to agent~$1$.
    However, such a partial allocation (and thus $\mathcal{A}$) will fail to be PO, as giving $o_1$ to agent~$1$ and $o_2$ to agent~$2$ instead will strictly increase the utility of both agents.

    \item[Case 2: $v(A^{2n+1}_1) < v(A^{2n+1}_2)$.]
    If we give $o_1$ to agent~$2$, since by (\ref{eqn:tef1po_kcondition_goods_2}), $v_1(o_1) > \max_{g \in A^{2n+1}_2} v(g)$, we get that
    \begin{equation*}
        v_1(A_1^{2n+2}) = v(A_1^{2n+1}) < v(A_2^{2n+1}) = v_1(A^{2n+2} \setminus \{o_1\}),
    \end{equation*}
    and agent~$1$ will still envy agent~$2$ after dropping $o_1$ from agent~$2$'s bundle.
    Thus, we must give $o_1$ to agent~$1$.

    Next, if we give $o_2$ to agent~$1$, then since $v_2(o_2) > \max_{g \in O} v(g)$ and $v_2(o_2) > v_2(o_1)$, we have that
    \begin{align*}
        v_2(A_2^{2n+3}) & = v(A_2^{2n+1}) \\
        & \leq \frac{v(O) + \max_{g \in O} v(g)}{2} \quad (\text{by } (\ref{eqn:tef1po_kcondition_goods_4}))\\
        & < \frac{v(O) - \max_{g \in O} v(g)}{2} + v_1(o_1) \quad (\text{by } (\ref{eqn:tef1po_kcondition_goods_2})) \\
        & \leq v(A_1^{2n+1}) + v_1(o_1) \quad (\text{by } (\ref{eqn:tef1po_kcondition_goods_3}))\\
        & = v_2(A_1^{2n+3} \setminus \{o_2\}),
    \end{align*}
    and agent~$2$ will still envy agent~$1$ after dropping $o_2$ from agent~$1$'s bundle.
    Thus, we must give $o_2$ to agent~$2$.

    Now, if we give $o_3$ to agent~$2$, then since $v_1(o_3) >  \max_{g \in O} v(g)$ and $v_1(o_3) = v_1(o_2)$, we have that
    \begin{align*}
        v_1(A_1^{2n+4}) & = v(A_1^{2n+1}) + v_1(o_1)\\
        & < v(A_2^{2n+1}) - \varepsilon + v_1(o_1) \quad (\text{by } (\ref{eqn:tef1po_kcondition_goods_5}))\\
        & = v(A_2^{2n+1}) + v_1(o_2)\\
        & = v_1(A_2^{2n+4} \setminus \{o_3\}),
    \end{align*}
    and agent~$1$ will still envy agent~$2$ after dropping $o_3$ from agent~$2$'s bundle.
    Thus, we must give $o_3$ to agent~$1$.

    Finally, if we give $o_4$ to agent~$1$, then since $v_2(o_3) > \max_{g \in O} v(g)$ and $v_2(o_3) > v_2(o_1) = v_2(o_4)$, we have that
    \begin{align*}
        v_2(A_2) & = v(A_2^{2n+1}) + v_2(o_2)\\
        & = v(A_2^{2n+1}) + 5^{m+n} \\
        & \leq \frac{v(O) + \max_{g \in O} v(g)}{2} + 5^{m+n} \quad (\text{by } (\ref{eqn:tef1po_kcondition_goods_4})) \\
        & < \frac{v(O) - \max_{g \in O} v(g)}{2} + 2 \times 5^{m+n} - 2\varepsilon \quad (\text{by } (\ref{eqn:tef1po_kcondition_goods_2})) \\
        & \leq v(A^{2n+1}_1) + 2\times5^{m+n} - 2\varepsilon \quad (\text{by } (\ref{eqn:tef1po_kcondition_goods_3}))\\
        & = v(A^{2n+1}_1) + v_2(\{o_1,o_4\}) \\
        & = v_2(A_1 \setminus \{o_3\}),
    \end{align*}
    and agent~$2$ will still envy agent~$1$ after dropping $o_3$ from agent~$1$'s bundle.
    Thus, we must give $o_4$ to agent~$2$.
    However, again, this is not PO as giving $o_3$ to agent~$2$ and $o_4$ to agent~$1$ will strictly increase the utility of both agents.
    \end{description}
    
    By exhaustion of cases, we have shown that if $v(A^{2n+1}_1) \neq v(A^{2n+1}_2)$, there does not exist a TEF1 and PO allocation over $O'$. Thus, a TEF1 and PO allocation over $O'$ exists if and only if $v(A^{2n+1}_1) \neq v(A^{2n+1}_2)$. By Lemma~\ref{lem:potef1goodshard}, this implies that a TEF1 and PO allocation over $O'$ exists if and only if there is a truth assignment $\alpha$ such that each clause in $F$ has exactly one \texttt{True} literal.

\subsection{Proof of \Cref{thm:tef1_po_nphard_chores}}
We reduce from the NP-hard problem \textsc{1-in-3-SAT}.
    An instance of this problem consists of conjunctive normal form $F$ with three literals per clause; it is a yes-instance if there exists a truth assignment to the variables such that each clause has exactly one \texttt{True} literal, and a no-instance otherwise.

    Consider an instance of \textsc{1-in-3-SAT} given by the CNF $F$ which contains $n$ variables $\{x_1,\dots,x_n\}$ and $m$ clauses $\{C_1,\dots,C_m\}$.

    We construct an instance $\mathcal{I}$ with two agents and $2n+1$ chores.
    For each $i \in [n]$, we introduce two chores $t_i, f_i$.
    We also introduce an additional chore $r$.
    Let agents' (identical) valuations be defined as follows:
    \begin{equation*}
        v(c) = 
        \begin{cases} 
            -5^{m+n-i} - \sum_{j \, : \, x_i \in C_j} -5^{m-j}, & \text{if } c = t_i, \\ 
            - 5^{m+n-i} - \sum_{j \, : \, \neg x_i \in C_j}  5^{m-j}, & \text{if } c = f_i, \\
            - \sum_{j \in[m]} 5^{j-1}, & \text{if } c = r.
        \end{cases}
    \end{equation*}
    Intuitively, for each variable index $i \in [n]$, we associate with it a unique value $-5^{m+n-i}$.
    For each clause index $j \in [m]$, we also associate it with a unique number $-5^{m-j}$.
    Note that no two indices (regardless of whether its a variable or clause index) share the same value, hence the term unique value.
    Then, the value for each chore $t_i$ comprises of the unique value associated with $i$, and the sum over all unique values of clauses $C_j$ which $x_i$ appears as a \emph{positive literal} in; whereas the value for each chore $f_i$ comprises of the unique value associated with $i$, and the sum over all unique values of clauses $C_j$ which $x_i$ appears as a \emph{negative literal} in.
    We will utilize this in our analysis later.

    Then, we have that the set of chores $O = \{t_1,f_1,t_2,f_2,\dots,t_n,f_n,r\}$.
    Note that
    \begin{equation*}
        v(O) = v(r) + \sum_{i \in [n]} v(t_i) + \sum_{i \in [n]} v(f_i).
    \end{equation*}
    Also observe that
    \begin{equation*}
        -\sum_{i \in [n]} 5^{m+n-i} = -\sum_{i \in [n]} 5^{m+i-1}.
    \end{equation*}
    Now, as each clause contains exactly three literals,
    \begin{equation*}
        -\sum_{i \in [n]}\sum_{j : x_i \in C_j} 5^{m-j} - \sum_{i \in [n]}\sum_{j : \neg x_i \in C_j} 5^{m-j} = 3 \times - \sum_{j \in [m]} 5^{j-1}.
    \end{equation*}
    Then, combining the equations above, we get that
        \begin{equation} \label{eqn:lem_vO_chores}
            v(O) = 2 \times -\sum_{i \in [n]} 5^{m+i-1} + 4 \times -\sum_{j \in [m]} 5^{j-1}.
        \end{equation}

    Let the chores appear in the following order:
    \begin{equation*}
        t_1,f_1,t_2,f_2,\dots,t_n,f_n,r.
    \end{equation*}
    We first prove the following result.
    \begin{lemma}\label{lem:potef1choreshard}
        There exists a truth assignment $\alpha$ such that each clause in $F$ has exactly one \texttt{True} literal if and only if there exists an allocation $\mathcal{A} = (A_1,A_2)$ such that $v(A_1) = v(A_2)$ for instance $\mathcal{I}$.
    \end{lemma}
    \begin{proof}
        For the `if' direction, consider an allocation $\mathcal{A}$ such that $v(A_1) = v(A_2)$.
        Since agents have identical valuations, without loss of generality, let $r \in A_1$.
        Since $O = A_1 \cup A_2$ and $v(A_1) = v(A_2) = \frac{1}{2} v(O)$, we have that
        \begin{equation*}
            v(A_1 \setminus \{r\}) = \left(- \sum_{i \in [n]} 5^{m+i-1} + 2 \times -\sum_{j \in [m]} 5^{j-1} \right) + \sum_{j \in [m]} 5^{j-1} = - \sum_{i \in [n]} 5^{m+i-1} -\sum_{j \in [m]} 5^{j-1}.
        \end{equation*}
        Note that this is only possible if for each $I \in [m]$, $t_i$ and $f_i$ are allocated to different agents.
        The reason is because the only way agent~$1$ can obtain the first term of the above bundle value (less chore $r$) is if she is allocated exactly one chore from each of $\{t_i,f_i\}$ for each $i \in [n]$.

        Then, from the chores that exists in bundle $A_1$, we can construct an assignment $\alpha$: for each $i \in [n]$, let $x_i = \texttt{True}$ if $t_i \in A_1$ and $x_i = \texttt{False}$ if $f_i \in A_1$.
        Then, from the second term in the expression of $v(A_1 \setminus \{r\})$ above, we can observe that each clause has exactly one \texttt{True} literal (because the sum is only obtainable if exactly one literal appears in each clause, and our assignment will cause each these literals to evaluate \texttt{True}.

        For the `only if' direction, consider a truth assignment $\alpha$ such that each clause in $F$ has exactly one \texttt{True} literal.
        Then, for each $i \in [n]$, let 
        \begin{equation*}
            \ell_i = 
            \begin{cases} 
                t_i & \text{if } x_i = \texttt{True} \text{ under } \alpha, \\
                f_i  & \text{if } x_i = \texttt{False} \text{ under } \alpha. 
            \end{cases}
        \end{equation*}
        We construct the allocation $\mathcal{A} = (A_1,A_2)$ where 
        \begin{equation*}
            A_1 = \{\ell_1, \dots, \ell_n, r\} \quad \text{and} \quad A_2 = O \setminus A_1.
        \end{equation*}
        Again, observe that
        \begin{equation*}
            -\sum_{i \in [n]} 5^{m+n-i} = -\sum_{i \in [n]} 5^{m+i-1}.
        \end{equation*}
        Then, as each clause has exactly one \texttt{True} literal, 
        \begin{equation*}
            v(A_1) = -\sum_{i \in [n]}  5^{m + i -1} + 2 \times -\sum_{j \in [m]} 5^{j - 1},
        \end{equation*}
        and together with (\ref{eqn:lem_vO_goods}), we get that
        \begin{equation*}
            v(A_2) = v(O) - v(A_1) = v(A_1),
        \end{equation*}
        as desired.
    \end{proof}
    Note that for all values of $m,n \geq 1$, and some $\varepsilon < \frac{1}{3}$,
    \begin{equation}\label{eqn:tef1po_kcondition_chores}
        \frac{-5^{m+n}+2\varepsilon}{2} < -5^{m+n-1} - \frac{5^m-1}{4} = -5^{m+n-1} - \sum_{j \in [m]} 5^{j-1} \leq \min_{c \in O} v(c).
    \end{equation}
    
    Now, consider another instance $\mathcal{I}'$ that is similar to $\mathcal{I}$, but with an additional four chores $o_1,o_2,o_3,o_4$.
    Let agents' valuations over these four new chores be defined as follows, for some $\varepsilon < \frac{1}{3}$:
    \begin{center}
        \begin{tabular}{l||cccc}
            $\mathbf{v}$ &$o_1$ & $o_2$ & $o_3$ & $o_4$\\ \hline\hline
            1 & $-5^{m+n}$ & $-5^{m+n} + \varepsilon$ & $-5^{m+n} + \varepsilon$ & $-5^{m+n}$\\
            2 & $-5^{m+n}+\varepsilon$ & $-5^{m+n}$ & $-5^{m+n}$ & $-5^{m+n}+\varepsilon$\\
        \end{tabular}
    \end{center}
    Then, we have the set of chores $O' = O \cup\{o_1,o_2,o_3,o_4\}$.
    
    Let the chores be in the following order:
    \begin{equation*}
        t_1,f_1,t_2,f_2,\dots,t_n,f_n,r, o_1,o_2,o_3,o_4.
    \end{equation*}
    If there is a partial allocation $\mathcal{A}^{2n+1}$ over the first $2n+1$ chores such that $v(A^{2n+1}_1) = v(A^{2n+1}_2)$, then by giving $o_1,o_4$ to agent 1 and $o_2,o_3$ to agent 2, we obtain an allocation that is TEF1 and PO (note that any allocation for the first $2n+1$ chores will be PO, since agents have identical valuations over them). 
    
    However, if there does not exist a partial allocation $\mathcal{A}^{2n+1}$ over the first $2n+1$ goods such that $v(A^{2n+1}_1) = v(A^{2n+1}_2)$, then let $\mathcal{A}^{2n+1}$ be any partial allocation of the first $2n+1$ goods that is TEF1 but $v(A^{2n+1}_1) \neq v(A^{2n+1}_2)$.

    Note that in order for $\mathcal{A}^{2n+1}$ to be TEF1, we must have that for any agent $i \in \{1,2\}$,
    \begin{equation} \label{eqn:tef1po_kcondition_chores_2}
        v(A^{2n+1}_i) - \min_{c \in O} v(c) \geq \frac{v(O)}{2}.
    \end{equation}
    This also means that for any agent $i \in \{1,2\}$,
    \begin{equation}\label{eqn:tef1po_kcondition_chores_3}
        v(A^{2n+1}_i) \leq v(O) - \left(\frac{v(O)}{2} +  \min_{c \in O} v(c)\right) = \frac{v(O)}{2} - \min_{c \in O} v(c).
    \end{equation}
    Also observe that since $\min_{g \in O} v(g) > \varepsilon$ and $v(A^{2n+1}_1) \neq v(A^{2n+1}_2)$, 
    \begin{equation}\label{eqn:tef1po_kcondition_chores_4}
        \left| v(A^{2n+1}_1) - v(A^{2n+1}_2) \right| > \varepsilon.
    \end{equation}

    We split our analysis into two cases.

    \begin{description}
        \item[Case 1: $v(A^{2n+1}_1) > v(A^{2n+1}_2)$.] 
    If we give $o_1$ to agent~$2$, since by (\ref{eqn:tef1po_kcondition_chores}), $v_2(o_1) < \min_{c \in A^{2n+1}_2} v(c)$, we get that
    \begin{equation*}
        v_2(A^{2n+2}_2 \setminus \{o_1\}) = v(A^{2n+1}_2) <  v(A^{2n+1}_1) = v_2(A^{2n+2}_1),
    \end{equation*}
    and agent~$2$ will still envy agent~$1$ after dropping $o_1$ from his own bundle.
    Thus, we must give $o_1$ to agent~$1$.

    Next, if we give $o_2$ to agent~$1$, then since $v_1(o_1) < \max_{c \in O} v(c)$ and $v_1(o_1) < v_1(o_2)$, we have that
    \begin{align*}
        v_1(A_1^{2n+3} \setminus \{o_1\}) & = v(A_1^{2n+1}) + v_1(o_2) \\
        & \leq \frac{v(O)}{2} - \min_{c \in O} v(c) + v_1(o_2) \quad (\text{by } (\ref{eqn:tef1po_kcondition_chores_3}))\\
        & < \frac{v(O)}{2} + \min_{c \in O} v(c) \quad (\text{by } (\ref{eqn:tef1po_kcondition_chores})) \\
        & \leq v(A_2^{2n+1}) \quad (\text{by } (\ref{eqn:tef1po_kcondition_chores_2}))\\
        & = v_1(A_2^{2n+3}),
    \end{align*}
    and agent~$1$ will still envy agent~$2$ after dropping $o_2$ from her own bundle.
    Thus, we must give $o_2$ to agent~$2$.
    However, such a partial allocation (and thus $\mathcal{A}$) will fail to be PO, as giving $o_1$ to agent~$2$ and $o_2$ to agent~$1$ will strictly increase the utility of both agents.

    \item[Case 2: $v(A^{2n+1}_1) < v(A^{2n+1}_2)$.]
    If we give $o_1$ to agent~$1$, since by (\ref{eqn:tef1po_kcondition_chores}), $v_1(o_1) \min_{c \in A^{2n+1}_1} v(c)$, we get that
    \begin{equation*}
        v_1(A_1^{2n+2} \setminus \{o_1\}) = v(A_1^{2n+1}) < v(A_2^{2n+1}) = v_1(A^{2n+2}),
    \end{equation*}
    and agent~$1$ will still envy agent~$2$ after dropping $o_1$ from her own bundle.
    Thus, we must give $o_1$ to agent~$2$.

    Next, if we give $o_2$ to agent~$2$, then since $v_2(o_2) < \min_{c \in O} v(c)$ and $v_2(o_2) < v_2(o_1)$, we have that
    \begin{align*}
        v_2(A_2^{2n+3} \setminus \{o_2\}) & = v(A_2^{2n+1}) + v_2(o_1)\\
        & \leq \frac{v(O)}{2} - \min_{c \in O} v(c) + v_2(o_1) \quad (\text{by } (\ref{eqn:tef1po_kcondition_chores_3}))\\
        & < \frac{v(O)}{2} + \min_{c \in O} v(c) \quad (\text{by } (\ref{eqn:tef1po_kcondition_chores})) \\
        & \leq v(A_1^{2n+1}) \quad (\text{by } (\ref{eqn:tef1po_kcondition_chores_2}))\\
        & = v_2(A_1^{2n+3}),
    \end{align*}
    and agent~$2$ will still envy agent~$1$ after dropping $o_2$ from his own bundle.
    Thus, we must give $o_2$ to agent~$1$.

    Now, if we give $o_3$ to agent~$1$, then since $v_1(o_3) < \min_{c \in O} v(c)$ and $v_1(o_3) = v_1(o_2)$, we have that
    \begin{align*}
        v_1(A_1^{2n+4} \setminus \{o_3\}) & = v(A_1^{2n+1}) + v_1(o_2)\\
        & < v(A_2^{2n+1}) - \varepsilon + v_1(o_2) \quad (\text{by } (\ref{eqn:tef1po_kcondition_chores_4}))\\
        & = v(A_2^{2n+1}) + v_1(o_1)\\
        & = v_1(A_2^{2n+4}),
    \end{align*}
    and agent~$1$ will still envy agent~$2$ after dropping $o_3$ from her own bundle.
    Thus, we must give $o_3$ to agent~$2$.

    Finally, if we give $o_4$ to agent~$2$, then since $v_2(o_3) < \min_{c \in O} v(c)$ and $v_2(o_3) < v_2(o_1) = v_2(o_4)$, we have that
    \begin{align*}
        v_2(A_2 \setminus \{o_3\}) & = v(A_2^{2n+1}) + v_2(\{o_1,o_4\})\\
        & = v(A_2^{2n+1}) - 2 \times 5^{m+n} + 2\varepsilon \\
        & \leq \frac{v(O)}{2} - \min_{c \in O} v(c) - 2\times 5^{m+n} + 2\varepsilon \quad (\text{by } (\ref{eqn:tef1po_kcondition_chores_3})) \\
        & < \frac{v(O)}{2} + \min_{c \in O} v(c) - 5^{m+n} + \varepsilon \quad (\text{by } (\ref{eqn:tef1po_kcondition_chores})) \\
        & \leq v(A^{2n+1}_1) + v_2(o_1) \quad (\text{by } (\ref{eqn:tef1po_kcondition_chores_2}))\\
        & = v(A^{2n+1}_1) + v_2(o_1) \\
        & = v_2(A_1),
    \end{align*}
    and agent~$2$ will still envy agent~$1$ after dropping $o_3$ from his own bundle.
    Thus, we must give $o_4$ to agent~$1$.
    However, again, this is not PO as giving $o_3$ to agent~$1$ and $o_4$ to agent~$2$ will strictly increase the utility of both agents.
    \end{description}
    By exhaustion of cases, we have shown that if $v(A^{2n+1}_1) \neq v(A^{2n+1}_2)$, there does not exist a TEF1 and PO allocation over $O'$. Thus, a TEF1 and PO allocation over $O'$ exists if and only if $v(A^{2n+1}_1) \neq v(A^{2n+1}_2)$. By Lemma~\ref{lem:potef1choreshard}, this implies that a TEF1 and PO allocation over $O'$ exists if and only if there is a truth assignment $\alpha$ such that each clause in $F$ has exactly one \texttt{True} literal.

\section{TEF1 for Mixed Manna} \label{app:mixed_manna}
    We first define TEF1 for mixed manna.
    \begin{definition}[Temporal EF1 for mixed manna]
    In the case of with both goods and chores, an allocation $\mathcal{A}^t = (A^t_1, \dots, A^t_n)$ is said to be \emph{temporal envy-free up to one item (TEF1)} if for all $t'\leq t$ and $i,j \in N$, there exists an item $o \in  A_i^{t'} \cup A_j^{t'}  $ such that $v_i(A_i^{t'} \setminus \{o\}) \geq v_i(A_j^{t'} \setminus \{o\})$. 
 \end{definition}
 Then, we can extend the result of \Cref{thm:2agents} to the more general mixed manna setting, with the following result.

\begin{theorem}\label{thm:mixedmanna}
    When $n=2$, a \emph{TEF1} allocation exists in the mixed manna setting, and can be computed in polynomial time.
\end{theorem}
    \begin{proof}
    For an agent $i\in \{1,2\}$ and round $t\in [T]$, we define $S^t_i \subseteq O^t$ as the set of items that have arrived up to round $t$ which only agent $i$ has a positive value for. 
    Then, for any $t\in [t]$ and $i,j\in \{1,2\}$ where $i\neq j$, $v_i(S_i^t)\geq 0$ and $v_i(S_j^t)\leq 0$. Clearly, if some allocation $\mathcal{A}^t$ is TEF1 over $O^t \setminus (S^t_1 \cup S^t_2)$, then $\mathcal{B}^t = (S^t_1 \cup A^t_1, S^t_2 \cup A^t_2)$ is a TEF1 allocation over $O^t$. Furthermore, for any $t \in [T]$ and $i,j\in \{1,2\}$ where $i\neq j$, if there exists an item $o \in  A_i^t \cup A_j^t  $ such that $v_i(A_i^t \setminus \{o\}) \geq v_i(A_j^t \setminus \{o\})$, then
    \begin{align*}
    v_i(B_i^t \setminus \{o\}) &= v_i(A_i^t \setminus \{o\}) + v_i(S_i^t)\\
    &\geq  v_i(A_i^t \setminus \{o\})\\
    &\geq v_i(A_j^t \setminus \{o\})\\
    &\geq  v_i(A_j^t \setminus \{o\}) +  v_i(S_j^t)\\
    &= v_i(B_j^t \setminus \{o\}),
    \end{align*}
    where the first and third inequalities are due to the fact that $ v_i(S_i^t) \geq 0$ and $ v_i(S_j^t) \leq 0$. It therefore suffices to assume that for each item $o \in O$, either $v_1(o)\leq 0$ and $v_2(o)\leq 0$, or $v_1(o)\geq 0$ and $v_2(o)\geq 0$, and we make this assumption for the remainder of the proof.

    Let $v'_i(o) = |v_i(o)|$ for all $i \in \{1,2\}$ and $o \in O$. Note that $v'_i(o) \geq 0$ for all $i \in \{1,2\}$ and $o \in O$ and thus, with respect to the augmented valuations, each $o \in O$ is a good. We use Algorithm 2 in \citet{he2019fairerfuturepast}, which returns a TEF1 allocation for goods in polynomial time, to compute an allocation $\mathcal{B}$ which is TEF1 with respect to the augmented valuations $\mathbf{v}'$. 

     For a round $t\in [T]$, let $G^t,C^t \subseteq O^t$ be, respectively, the subsets of goods and chores (with respect to the original valuation profile  $\mathbf{v}=(v_1,v_2)$) that have arrived up to round $t$. Then, for each $t\in [T]$ and $i \in \{1,2\}$, let $G^t_i = G^t \cap B^t_i$ and $C^t_i = C^t \cap B^t_i$. We construct allocation  $\mathcal{A}= (G^T_1 \cup C^T_2, G^T_2 \cup C^T_1)$ from $\mathcal{B}$ by swapping the agents' bundles of chores. We now show that $\mathcal{A}$ is TEF1. 

Recall that all items are goods with respect to $\mathbf{v}'$. Since $\mathcal{B}$ is TEF1, we know that for any $t\in [T]$ and $i,j \in \{1,2\}$ where $i\neq j$, there exists an item $o \in B_j^t$ such that 
\begin{align*}
    &v'_i(B_i^t) \geq v'_i(B_j^t \setminus \{o\}) \\
    &\implies v_i(G_i^t) - |v_i(C_i^t)| \geq v_i(G_j^t) - |v_i(C_j^t)| - |v_i(o)| \\
     &\implies v_i(G_i^t) + v_i(C_j^t) \geq v_i(G_j^t) + v_i(C_i^t) - |v_i(o)| \\
     & \implies \begin{cases}
         v_i(G_i^t) + v_i(C_j^t) \geq v_i(G_j^t \setminus\{o\}) + v_i(C_i^t)& \text{if $o \in G^t_j$}\\
      v_i(G_i^t) + v_i(C_j^t \setminus\{o\})  \geq v_i(G_j^t) + v_i(C_i^t) & \text{if $o \in C^t_j$}
     \end{cases}\\
     &\implies v_i(G_i^t \cup C_j^t \setminus\{o\}) \geq v_i(G_i^t \cup C_j^t \setminus\{o\}) \\
     &\implies  v_i(A_i^t  \setminus\{o\}) \geq v_i(A_j^t  \setminus\{o\}).
\end{align*}
Thus, $\mathcal{A}$ is TEF1. 
\end{proof}
\end{document}